\documentclass[journal,twocolumn,10pt]{IEEEtranTCOM}
\usepackage{color}
\usepackage{amsfonts,amsmath,amssymb,amsthm,mathrsfs,bm}
\usepackage{latexsym,amscd}
\usepackage{amsbsy}
\usepackage{graphicx,multirow,bm}
\graphicspath{{20260511/}{./}}
\usepackage{algorithm}
\usepackage{algorithmicx}
\usepackage{algpseudocode}
\usepackage{xcolor}
\usepackage{tikz}
\usepackage{diagbox}
\usepackage{multicol}
\usepackage{arydshln}
\usepackage{booktabs}
\usepackage{times}
\usepackage{cuted}
\usepackage{cases}
\usepackage{subfigure}
\usepackage{url}
\usepackage[noadjust]{cite}
\def\BibTeX{{\rm B\kern-.05em{\sc i\kern-.025em b}\kern-.08em
    T\kern-.1667em\lower.7ex\hbox{E}\kern-.125emX}}
\usetikzlibrary{arrows.meta,positioning,calc,shapes.geometric}

\definecolor{bluebox}{RGB}{35,100,200}
\definecolor{orangebox}{RGB}{230,120,0}
\definecolor{greenbox}{RGB}{70,140,60}
\definecolor{redbox}{RGB}{220,80,35}
\definecolor{readblue}{RGB}{95,150,220}
\definecolor{newgreen}{RGB}{130,210,100}
\definecolor{stripebg}{RGB}{238,238,238}
\definecolor{unchangedgray}{RGB}{145,145,145}
\definecolor{darkgray}{RGB}{50,50,50}

\tikzset{
    panel/.style={
        rounded corners=3pt,
        line width=0.8pt,
        fill=#1!4
    },
    paneltitle/.style={
    font=\bfseries\scriptsize,
        align=center
    },
    bodytext/.style={
        font=\tiny,
        align=center
    },
    tinybody/.style={
    font=\fontsize{5}{6}\selectfont,
        align=center
    },
    flowarrow/.style={
        -{Triangle[length=4mm,width=4mm]},
        line width=2.2pt,
        draw=gray!70
    },
    smallarrow/.style={
        -{Triangle[length=3mm,width=3mm]},
        line width=1.1pt,
        draw=bluebox
    },
    dashbox/.style={
        rounded corners=3pt,
        dashed,
        line width=0.7pt,
        draw=gray!80,
        fill=stripebg
    },
    stripebox/.style={
    rounded corners=3pt,
    dashed,
    line width=0.7pt,
    draw=gray!75,
    fill=stripebg
},
    symbol/.style={
    rectangle,
    minimum width=0.31cm,
    minimum height=0.31cm,
    inner sep=0pt,
    draw=darkgray,
    line width=0.40pt
},
read/.style={
    symbol,
    fill=readblue
},
unchanged/.style={
    symbol,
    fill=unchangedgray
},
newsym/.style={
    symbol,
    fill=newgreen
},
midarrow/.style={
    -{Triangle[length=3.2mm,width=3.2mm]},
    line width=1.6pt,
    draw=gray!70
},
cyl/.style={
    cylinder,
    shape border rotate=90,
    aspect=0.35,
    minimum width=0.18cm,
    minimum height=0.30cm,
    inner sep=2.8pt,
    outer sep=0pt,
    draw=darkgray,
    line width=0.35pt,
    fill=#1
}
}

\newtheorem{thm}{Theorem}
\newtheorem{lem}{Lemma}

\newtheorem{defn}{Definition}
\newtheorem{ex}{Example}
\newtheorem{remark}{Remark}

\newtheorem{cons}{Construction}

\newcommand{\fq}{{\mathbb F}_{q}}

\renewcommand\thetable{\Roman{table}}

\usepackage{pgf}
\usepackage{tikz}
\usetikzlibrary{matrix,chains,positioning,arrows,calc,decorations,plotmarks,patterns,trees,fit,backgrounds,shapes.multipart,topaths}
\usepackage{pgfplots}
\usetikzlibrary{external}                       
\usepackage{scalefnt}
\usepackage{tkz-graph}
\pgfplotsset{compat=1.3}
\tikzstyle{help lines}=[black!20,dashed]

\pgfplotscreateplotcyclelist{mylist}{red,blue,black,yellow,brown}

\pgfdeclarepatternformonly{my north east lines}{\pgfqpoint{-1pt}{-1pt}}{\pgfqpoint{8pt}{8pt}}{\pgfqpoint{6pt}{6pt}}%
{
	\pgfsetlinewidth{0.4pt}
	\pgfpathmoveto{\pgfqpoint{0pt}{0pt}}
	\pgfpathlineto{\pgfqpoint{6.1pt}{6.1pt}}
	\pgfusepath{stroke}
}
\definecolor{light_gray}{rgb}{0.6,0.6,0.6}
\definecolor{awgray}{rgb}{0.7,0.7,0.7}
\definecolor{awgray_dark}{rgb} {0.4,0.4,0.4}

\tikzset{
	>=stealth',
	mycircle/.style={circle, draw=gray, very thick},
	mycircle_small/.style={circle,draw=awgray_dark,fill = awgray_dark, inner sep=0,minimum size=.6em},
	mycircle_small_black/.style={circle,draw=black,fill = black, inner sep=0,minimum size=.6em},
	mybox/.style={rectangle,rounded corners,draw=black, thick,text width=1em,minimum height=4em,minimum width=5em,text centered},
	mybox_big/.style={rectangle,rounded corners,draw=black, thick,text width=17.5em,minimum height=12em,minimum width=17.5em,text centered},
	mybox_vec/.style={rectangle,rounded corners,draw=black, thick,text width=4em,minimum height=0.7em, minimum width=4em,text centered},
	mybox_vec_short/.style={rectangle,rounded corners,draw=black, thick,text width=1em,minimum height=0.7em, minimum width=2em,text centered},
	pil/.style={->, thick, shorten <=2pt, shorten >=2pt,},
}

\usetikzlibrary{arrows,shapes,chains}
\renewcommand\baselinestretch{1}

\begin{document}
	
	\title{Convertible Codes: MSR-to-MSR Conversion with Optimal Access and Bandwidth}
	\author{Yumeng Yang,~\IEEEmembership{Student Member,~IEEE}, Han Cai,~\IEEEmembership{Member,~IEEE},  Xianfu Lei, \IEEEmembership{Member,~IEEE}, and Xiaohu Tang, \IEEEmembership{Fellow, IEEE} 
		
    \thanks{
			Y. Yang, H. Cai, X. Lei and X. Tang are with the Information Coding and Transmission Key Laboratory of Sichuan Province, Southwest Jiaotong University, Chengdu 610032, China (email: yangyumeng@my.swjtu.edu.cn, hancai@swjtu.edu.cn, xflei@swjtu.edu.cn, xhutang@swjtu.edu.cn).
}
}

	\maketitle
\vspace{-2cm}
\begin{abstract}

In this paper, we study convertible codes in the merge regime and focus on the minimum storage regenerating (MSR) setting, 
where both the initial codes and the final code admit optimal single-node repair. We propose explicit MSR-to-MSR conversion 
schemes and analyze their performance in terms of access cost and conversion bandwidth.

We first construct convertible MSR codes in the irregular setting, where the $m$ initial codes may have different parameters, achieving optimal access cost. We further consider the practically important same-code setting, where all initial codewords are drawn from the same MSR code. By introducing a row-matching technique, we obtain constructions simultaneously achieving optimal access cost and conversion 
bandwidth in most parameter regimes.

\end{abstract}


	\section{Introduction}

Distributed storage systems provide reliability by encoding data 
across multiple storage nodes using erasure codes. Two important 
and practically motivated aspects of such systems have been 
extensively studied: efficient node repair to handle individual 
node failures, and efficient code conversion to adapt to changing 
system requirements.

The problem of node repair is addressed by regenerating codes \cite{DGW+10}, which minimize the repair bandwidth by 
allowing each helper node to transmit a fraction of its stored 
data. Among these, Minimum Storage Regenerating (MSR) codes are 
of particular interest because they achieve the minimum possible 
storage overhead and optimal repair bandwidth. Explicit 
constructions of MSR codes have been proposed 
in~\cite{LY,LJ,LTP15,ZLH25,LWH24,WZLT23,ZZ23,HLPS+19,ChenBarg20,WC26,Vajha23,Ye2017Explicit,Ye2020,tian2013,Zorgur2019,tamo2019}.

Separately, as storage requirements evolve---due to data growth, 
node addition, or changes in fault-tolerance requirements---it 
becomes necessary to adapt the storage configuration without 
re-encoding from scratch. This problem is addressed by convertible 
codes~\cite{maturana2022Convertible}, 
which enable the transformation from an initial code to a final 
code with different parameters while accessing only a subset of stored symbols. The performance of such schemes is measured by 
the \emph{access cost}, i.e., the total number of symbols read 
during conversion, and the \emph{conversion bandwidth}, i.e., 
the total amount of data downloaded across all nodes.
Information-theoretic lower bounds for 
these metrics have been established for scalar MDS codes 
in~\cite{ge2024MDS} and \cite{maturana2023Bandwidth}. 
The most well-studied setting is the \emph{merge} 
regime~\cite{chopra2024Low,ge2024MDS,ge2025Locally,Kong2024LRC,maturana2022Convertible,shi2025Bounds}, where 
$m$ initial codewords are combined into a single final codeword 
of larger dimension, which is the setting we consider throughout 
this paper.

Although MSR codes achieve optimal repair bandwidth and convertible codes
enable efficient transformation between different parameters,
their integration remains challenging. Existing convertible-code
constructions based on MDS codes  do not preserve the repair
structure, while MSR codes are designed for static settings and
do not support efficient conversion. The main difficulty lies in a structural mismatch:  code conversion
reorganizes stripe layout and symbol placement, whereas MSR
repair relies on a coordinate-wise structure enabling interference
alignment. As a result, this structure is not preserved under
conversion, making it difficult to achieve both efficient
conversion and optimal repair simultaneously.

This motivates the study of \emph{convertible MSR codes}, which 
support both efficient repair and efficient conversion. The main 
contributions are as follows.

First, we extend the scalar MDS conversion framework of~\cite{ge2024MDS} to array codes, providing a general foundation for our convertible MSR constructions.

Then, we construct explicit $(m,1)_q$ convertible MSR codes in the irregular setting that achieve optimal access cost, based on this framework and the Hadamard-design MSR construction~\cite{Ye2017Explicit}. The key idea is a subsymbol-level pre-alignment that preserves the periodic structure required for MSR repair under conversion.

Finally, we further consider the same-code setting, where all initial codewords are drawn from the same MSR code. By introducing a row-matching technique based on bijective row remapping, we obtain constructions simultaneously achieving optimal access cost and conversion bandwidth for $r_F\leq\min\{k_I,r_I\}$ and $r_F>k_I$. For the remaining regime $r_I<r_F<k_I$, a new construction combining the conversion procedure with regenerating techniques achieves optimal conversion bandwidth.
 As a comparison, Table~\ref{tab: literature_comparison} presents 
parameters of previous results and ours. 
As shown in the table, existing constructions achieve 
at most two of the three desirable properties---optimal 
access cost, optimal conversion bandwidth, and optimal 
repair bandwidth---simultaneously. In contrast, the 
constructions proposed in this paper are the first to 
achieve all three properties at once, across different 
parameter regimes.

 This article substantially extends our preliminary conference paper~\cite{YangISIT2026} in several important respects, including more general parameter regimes, a more general construction framework, and strengthened optimality analysis.

\begin{table*}[t]
\caption{Comparison of representative prior $\mathrm{MDS}\to\mathrm{MDS}$ results in the merge regime. The conversion-bandwidth column is considered only in the same-code setting.}
\label{tab: literature_comparison}
\centering
\vspace{-0.3cm}
\scriptsize
\setlength{\tabcolsep}{3pt}
\renewcommand{\arraystretch}{1}
\begin{tabular}{p{2.7cm}p{2.1cm}p{4.7cm}p{1.35cm}p{1.5cm}p{1.2cm}}
\toprule
Construction & Setting & Parameters & Opt. access & Opt. bw. & Opt. repair \\
\midrule
General MDS~\cite{maturana2022Convertible}
& same-code
& $r_F\leq\min\{k_I,r_I\}$
& $\checkmark$
& $\times$
& $\times$ \\

Hankel-I~\cite{maturana2022Convertible}
& same-code
& $r_F\leq\lfloor r_I/m\rfloor$
& $\checkmark$
& $\times$
& $\times$ \\

Hankel-II~\cite{maturana2022Convertible}
& same-code
& $r_F\leq r_I-m+1$
& $\checkmark$
& $\times$
& $\times$ \\

Hankel$_s$~\cite{maturana2022Convertible}
& same-code
& for $m\leq s\leq r_I$,\; $r_F\leq (s-m+1)\lfloor r_I/s\rfloor+\max\{(r_I\bmod s)-m+1,0\}$
& $\checkmark$
& $\times$
& $\times$ \\

GRS~\cite{Kong2024LRC}
& same-code
& $r_F\leq\min\{k_I,r_I\}$
& $\checkmark$
& $\times$
& $\times$ \\

Extended GRS~\cite{ge2024MDS}
& irregular / same-code
& ---
& $\checkmark$
& $\times$
& $\times$ \\

Bandwidth-optimal MDS~\cite{maturana2023Bandwidth}
& same-code
& $r_F\leq\min\{k_I,r_I\}$
& $\checkmark$
& $\checkmark$
& $\times$ \\

Bandwidth-optimal MDS~\cite{maturana2023Bandwidth}
& same-code
& $r_I<r_F\leq k_I$
& $\times$
& $\checkmark$
& $\times$ \\

MDS~\cite{shi2025Bounds}
& irregular / same-code
& $m\leq r_F$
& $\checkmark$
& $\times$
& $\times$ \\

Construction~\ref{cons: array_rI>r_F} 
& irregular / same-code
& $r_F\leq\min\{k_{I_i},r_{I_i}\}$ for all $i\in[m]$ / $r_F\leq\min\{k_{I},r_{I}\}$
& $\checkmark$
& $\checkmark$
& $\checkmark$ \\

{Construction\,1A} (Remark~\ref{rmk: construction2})
& irregular / same-code
& $r_F>\min\{k_{I_i},r_{I_i}\}$ for all $i\in[m]$ / $r_F>k_I$
& $\checkmark$
& $\checkmark$
& $\checkmark$ \\

Construction~\ref{cons: array_bw} 
&  same-code
& $r_{I}<r_F<k_{I}$
& $\times$
& $\checkmark$
& $\checkmark$ \\

\bottomrule
\end{tabular}
\end{table*}

The remainder of the paper is organized as follows. 
Section~\ref{sec: pre} introduces notation and preliminaries. Section~\ref{sec:model} presents the system model and formal
definitions of convertible MSR codes.
Section~\ref{sec: framework} presents the array MDS conversion 
framework. Section~\ref{sec: array} presents the irregular and 
same-code constructions.
Section~\ref{sec: conclusion} concludes 
the paper.

\section{Preliminaries}\label{sec: pre}
An $(n,k;\alpha)_q$ array code $\mathcal{C}$ is a subspace of $\mathbb{F}_q^{n\alpha}$ with dimension $k\alpha$. For a given codeword $\mathbf{c}= (\mathbf{c}_0,\dots,\mathbf{c}_{n-1})\in\mathcal{C}$, denote by $\mathbf{c}_i=(c_{0,i},\dots,c_{\alpha-1,i})$ for $i\in[0,n-1]$ the $i$-th {\it symbol} of $\mathbf{c}$. Moreover, each {\it node} stores one vector symbol of length $\alpha$, and we call each component $c_{w,i}$ in this vector a {\it subsymbol} for $i\in[0,n-1]$ and $w\in[0,\alpha-1]$.

We use generalized Reed-Solomon codes as a base code to design convertible codes. The definition is given below.
\begin{defn}[Generalized Reed-Solomon code]\label{Def_RS}
Let $\Lambda=\{\lambda_{0},\cdots,\lambda_{n-1}\}\subseteq \mathbb{F}_q$ be the set of distinct elements. A generalized Reed-Solomon code of length $n$ and dimension $k$ with evaluation points $\Lambda$ is defined as
\begin{equation*}
\{(\nu_{0}f(\lambda_{0}),\cdots,\nu_{n-1}f(\lambda_{n-1})):f\in\mathbb{F}_q[x],{\rm deg}(f)<k\},
\end{equation*}
where $\boldsymbol{\nu}=(\nu_{0},\cdots,\nu_{n-1})\in(\mathbb{F}_q^{*})^n.$
\end{defn}

In particular, when $\boldsymbol{\nu}=(1,\ldots,1)$, the resulting code is the classical Reed--Solomon (RS) code. 
The dual of an RS code is also a generalized Reed--Solomon code. 
Specifically, for an $[n,k]$ RS code with evaluation points $\Lambda=\{\lambda_{0},\ldots,\lambda_{n-1}\}\subseteq \mathbb{F}_q$, the coefficients of its dual code are given by
\begin{equation}\label{eq: coefficient}
\nu_i = \prod_{\substack{j\in[0,n-1]\\ j\neq i}}(\lambda_i-\lambda_j)^{-1}, \qquad i\in[0,n-1].
\end{equation}

Clearly, the coefficients $\boldsymbol{\nu}$ are uniquely determined by the evaluation points.

We next recall the notion of puncturing.
\begin{defn}[Punctured code]
Let $\mathcal{C}$ be an $(n,k;\alpha)_q$ array code and let $\mathcal{P}\subseteq[0,n-1]$ be an index set with $|\mathcal{P}|=s$.
The \emph{punctured code} of $\mathcal{C}$ on $\mathcal{P}$ is
\[
\mathcal{C}|_{\mathcal{P}}:=\{\mathbf{c}|_{\mathcal{P}}:\mathbf{c}\in\mathcal{C}\}\subseteq \mathbb{F}_q^{s\alpha},
\]
where $\mathbf{c}|_{\mathcal{P}}:=(\mathbf{c}_i: i\in\mathcal{P})$ is obtained by keeping only the coordinates indexed by $\mathcal{P}$.
\end{defn}

Thus, for an $[n,k]$ GRS codeword $\mathbf{c}=(c_0,\cdots,c_{n-1})=(\nu_{0}f(\lambda_{0}),\cdots,\nu_{n-1}f(\lambda_{n-1}))$ with evaluation points $\Lambda=\{\lambda_{0},\cdots,\lambda_{n-1}\}$, a parity-check matrix can be taken as the $(n-k)\times n$ Vandermonde matrix $H=\mathrm{Vand}_{n-k}(\lambda_{0},\cdots,\lambda_{n-1})$, i.e.,
\[(H)_{t,i}=\lambda_i^t,\qquad t\in[0,n-k-1],\ i\in[0,n-1],\]
which satisfies
$H\cdot (c_0,\cdots,c_{n-1})^{\top}=0.$

In the sequel, when referring to puncturing $\mathbf{c}$, we adopt
a convenient normalization that preserves the Vandermonde
structure of the parity-check matrix.

Specifically, let $\mathcal{P}=\{0,\cdots,n'-1\}$ with $k<n'<n$,
and define
\[
H|_{\mathcal{P}} := \mathrm{Vand}_{r'}(\lambda_0,\cdots,\lambda_{n'-1}),
\quad r'=n'-k,
\]
by restricting the evaluation points to $\mathcal{P}$.

To ensure consistency with this reduced parity-check matrix,
the codeword symbols are rescaled according to the updated
GRS coefficients
\[
\nu'_i := \prod_{\substack{j\in[0,n'-1]\\ j\neq i}}
(\lambda_i-\lambda_j)^{-1}, \quad i\in[0,n'-1].
\]
Then the rescaled vector
$
\big(\nu'_0\nu_0^{-1}c_0,\cdots,\nu'_{n'-1}\nu_{n'-1}^{-1}c_{n'-1}\big)
$
satisfies the parity-check equations defined by $H|_{\mathcal{P}}$.

Therefore, without loss of generality, we assume this scaling
is applied locally by each node and treat it implicitly, and
represent the punctured codeword as $(c_0,\cdots,c_{n'-1})$
with parity-check matrix $H|_{\mathcal{P}}$.

\subsection{Minimum Storage Regenerating Code}

In distributed storage systems, regenerating codes are designed to efficiently repair failed storage nodes. 
A data file $\mathcal{M}$ is encoded and stored across $n$ storage nodes, and each node stores a vector of length $\alpha$, i.e. $\alpha$ subsymbols.
When a node fails, a replacement node is generated by downloading $b$ subsymbols from each of a subset of surviving nodes, referred to as {\it helper nodes.} The total amount of subsymbols $db$ downloaded from $d$ helper nodes is called {\it repair bandwidth.}

In this work, we focus on regenerating codes achieving the minimum storage point, which satisfy the MDS property. Such codes are referred to as minimum storage regenerating (MSR) codes.

\begin{defn}[Minimum storage regenerating code \cite{DGW+10}]
An $(n,k;\alpha)_q$ regenerating code is called a \emph{minimum storage regenerating (MSR) code} if it achieves the cut-set bound with equality, i.e.,
\begin{equation}\label{eq: beta}\small
\begin{split}
\alpha=\frac{\mathcal{M}}{k},
\quad
b = \frac{\alpha}{d-k+1},
\end{split}
\end{equation}
where $d \leq n-1$ is the number of helper nodes and $b$ denotes the number of subsymbols downloaded from each helper node during repair.
\end{defn}
For the general cut-set bound, readers may refer to \cite{DGW+10}.

\subsection{Hadamard-Design MSR Code}\label{sec: hadamard}

We briefly recall the Hadamard-design MSR construction in~\cite{Ye2017Explicit}, which supports arbitrary parameters $n$, $k$, and $d$, and will be used later in our constructions.

\textbf{1. Hadamard-design MSR construction:}
Consider an $(n,k;\alpha)$ MSR code over $\mathbb{F}_q$ with $q>sn$, and define
$
s:=d-k+1,\, \alpha:=s^n.
$
Let $(\mathbf{c}_0,\cdots,\mathbf{c}_{n-1})$ be a codeword, where $\mathbf{c}_i=(c_{0,i},\cdots,c_{\alpha-1,i})^{\top}$ for $i\in[0,n-1]$.

Select $sn$ distinct elements
\begin{equation}\label{eq: distinct points}
\{\lambda_{a,j}:a\in[0,s-1],\,j\in[0,n-1]\}\subseteq\mathbb{F}_q.
\end{equation}

The Hadamard MSR construction is based on structural properties that enable efficient repair. We describe it in two steps.

\textbf{Step (i) Establish the layered MDS property:}
For each $w\in[0,\alpha-1]$, choose elements $\{\alpha_{w,j}:j\in[0,n-1]\}\subseteq\fq$ such that
$
\alpha_{w,j}\neq\alpha_{w,j'}\quad\text{for }j\neq j'.
$

Define an array code by parity-check equations
\begin{equation}\label{eq: had_pc_row}
\sum_{j=0}^{n-1}\alpha_{w,j}^{t}c_{w,j}=0,\qquad t\in[0,n-k-1],\, w\in[0,\alpha-1].
\end{equation}
Clearly, each row is a scalar $[n,k]$ MDS code, and thus the full code is an $(n,k;\alpha)$ MDS array code.

\textbf{Step (ii) Impose coordinate-wise periodicity:}
Index each row by the $s$-ary expansion
$
w=(w_{(0)},\ldots,w_{(n-1)})\in[0,s-1]^n.
$

To endow the above MDS array code with the Hadamard structure, define
\begin{equation}\label{eq: had_row_points_simple}
\alpha_{w,j}:=\lambda_{w_{(j)},j},\qquad j\in[0,n-1].
\end{equation}
By \eqref{eq: distinct points}, we have $\lambda_{w_{(j)},j}\neq \lambda_{w_{(j')},j'}$ for $j\neq j'$.

\textbf{2. Optimal repair property:}
Intuitively, \eqref{eq: had_pc_row} guarantees row-wise MDS reliability, while \eqref{eq: had_row_points_simple} makes each node coefficient depend only on one row digit. This allows us to group rows that differ only in the failed coordinate, so helper interference aligns into one summed subsymbol per helper, whereas the failed-node subsymbols remain distinguishable and can be solved with optimal repair bandwidth.

\textbf{(i) Single-node repair:}
Fix a failed node $j^*\in[0,n-1]$. Let $D\subseteq[0,n-1]\backslash\{j^*\}$ with $|D|=d$ be the helper-node set. For each row index $w$, define the group
\[
\mathcal{G}(w,j^*):=\{R_w(j^*,a):a\in[0,s-1]\},
\]
where $R_w(j^*,a)$ is obtained by replacing the $j^*$-th coordinate of $w$ with $a$. Clearly, the collection of all distinct groups $\mathcal{G}(w,j^*)$ form a partition of $[0,\alpha-1]$.
By the coordinate-wise periodicity, we have
\begin{equation}\label{eq: periodicity}
\alpha_{R_w(j^*,a),j}=
\begin{cases}
\lambda_{w_{(j)},j}, & j\neq j^*,\\
\lambda_{a,j^*}, & j=j^*,
\end{cases}
\quad a\in[0,s-1].
\end{equation}
For each $t\in[0,n-k-1]$, summing the $s$ equations in \eqref{eq: had_pc_row} corresponding to all rows in the group $\mathcal{G}(w,j^*)$ gives
\[\small
\sum_{j\neq j^*}\lambda_{w_{(j)},j}^{t}\,\sum_{a=0}^{s-1}c_{R_w(j^*,a),j}
+\sum_{a=0}^{s-1}\lambda_{a,j^*}^{t}\,c_{R_w(j^*,a),j^*}=0,
\]
where \eqref{eq: periodicity} ensures that the coefficients for $j\neq j^*$ are invariant within each group.
Hence, the repair procedure is given as follows:
\begin{itemize}
\item \textbf{Download:} For each $j\in D$, the $j$-th helper node transmits one subsymbol $\sum_{a=0}^{s-1}c_{R_w(j^*,a),j}$ per row group;
\item \textbf{Repair:} The failed node $j^*$ solves $s$ unknowns $\{c_{R_w(j^*,a),j^*}:\,a\in[0,s-1]\}$ per group.  
\end{itemize}
Running over all groups recovers the whole symbol in node $j^*$,
which gives
$
b=\frac{\alpha}{s}=\frac{\alpha}{d-k+1},
$
meeting the cut-set bound in~\eqref{eq: beta}.

\textbf{(ii) Multiple-node repair:}
Compared with part (i), the parameter change here comes from applying the space-sharing technique
for multi-node repair. Specifically, we extend $d-k+h$ instances of the single-node repair
construction, so the row index is extended and the sub-packetization becomes
\(\alpha=(d-k+h)s^n\). In the multiple-node 
case, the repair is performed in a centralized manner that collects all 
downloaded data from helper nodes and recovers the failed symbols.
We consider the following repair scheme from \cite{Ye2020}. Although originally described for the cooperative repair model, it achieves the centralized cut-set bound once the cooperation phase is removed, as observed in \cite{Zhang2025}.

Index each subsymbol by $(w,z)$ with $w\in[0,s^n-1]$ and $z\in[0,d-k+h-1]$. Thus, each node $u\in[0,n-1]$ stores
\[
\{c_{w,z,u}: w\in[0,s^n-1],\ z\in[0,d-k+h-1]\},
\]
where $w=(w_{(0)},\cdots,w_{(n-1)})\in[0,s-1]^n$ is the $s$-ary expansion of $w$. 
Here $\oplus$ denotes addition modulo $s$.

Let the failed set be $\mathcal{F}=\{j_1,\cdots,j_h\}$ and $D\subseteq[0,n-1]\backslash\mathcal{F}$ with $|D|=d$ be the helper-node set. For each $w$ and $l\in[h]$, define
\begin{equation*}\small
\begin{split}
\mathcal{G}(w,j_l):=\{(R_w(j_l,&w_{(j_l)}\oplus a),\,a) : a\in[0,s-2]\}\cup\\
&\{(R_w(j_l,w_{(j_l)}\oplus (s-1)),\,s+l-2) \},
\end{split}
\end{equation*} 
where these groups form a  partition of $[0,\alpha-1]$.

For each $t\in[0,n-k-1]$, summing all the $s$ parity-check equations corresponding to all rows in group $\mathcal{G}(w,j_l)$ gives
\begin{equation}\label{eq: multi_node equation}\small
\begin{split}
&\sum_{a=0}^{s-2}\lambda_{w_{(j_l)}\oplus a,j_l}^{t}\,c_{R_w(j_l,w_{(j_l)}\oplus a),a,j_l}\\
+&\lambda_{w_{(j_l)}\oplus (s-1),j_l}^{t}\,c_{R_w(j_l,w_{(j_l)}\oplus (s-1)),s+l-2,j_l}\\
+&\sum_{j\neq j_l}\lambda_{w_{(j)},j}^{t}\,\big(\sum_{a=0}^{s-2}c_{R_w(j_l,w_{(j_l)}\oplus a),\,a,j}\\
+&c_{R_w(j_l,w_{(j_l)}\oplus (s-1)),s+l-2,j}\big)=0.
\end{split}
\end{equation}
Then, by \eqref{eq: multi_node equation}, the repair procedure is given as follows:
\begin{itemize}
\item \textbf{Download:} For $l\in[h]$ and a group $\mathcal{G}(w,j_l)$, each helper node $j\in D$ transmits one group-sum
\[\small
\sum_{a=0}^{s-2}c_{R_w(j_l,w_{(j_l)}\oplus a),a,j}+c_{R_w(j_l,w_{(j_l)}\oplus s-1),s+l-2,j},
\]
and the $d$ helper-sums can determine $d-k+h$ unknowns in this group.
\item \textbf{Repair:} The difference from single-node repair is the composition of these unknowns: each group $\mathcal{G}(w,j_l)$ recovers exactly $d-k+h$ subsymbols, where
\begin{itemize}
\item $d-k+1$ subsymbols corresponding to failed node $j_l$: 
\begin{equation*}
\begin{split}
  \{c_{R_w(j_l,w_{(j_l)}\oplus a),a,j_l}: \,&a\in[0,s-2]\}\\&\cup\{c_{R_w(j_l,w_{(j_l)}\oplus s-1),s+l-2,j_l}\}.
  \end{split}
  \end{equation*}
\item $h-1$  subsymbols corresponding to failed nodes $j_{l'}$, $l'\in[h]\backslash\{ l\}$:
\begin{equation*}
\begin{split}
\{\sum_{a=0}^{s-2}c_{R_w(j_{l},w_{(j_{l})}\oplus a),a,j_{l'}}+&c_{R_w(j_{l},w_{(j_{l})}\oplus s-1),s+{l}-2,j_{l'}}:\\&l'\in[h]\setminus\{l\}\}.
\end{split}
\end{equation*}

\end{itemize}
\end{itemize}
Running over all $(w,l)$  and combining all the repaired subsymbols recovers all symbols in $\{j_1,\cdots,j_h\}$.

Each helper sends one symbol per group, i.e., $hs^n$ subsymbols in total. Hence the per-helper repair bandwidth is
$
b=hs^n=\frac{h\alpha}{d-k+h},
$
which meets the centralized multi-node cut-set bound in \cite{Zorgur2019}.

\section{System Model and Convertible MSR Codes}
\label{sec:model}

This section presents the system model and formalizes
convertible MSR codes. We focus on two key abstractions of system reconfiguration:
the cost of code conversion and the repair efficiency after conversion.
Fig.~\ref{fig:convertible-msr} illustrates the system workflow,
and Fig.~\ref{fig:reconfiguration-b} provides a symbol-level
view in the merge regime.

\begin{figure*}[!t]
    \centering
    \includegraphics[width=0.85\textwidth]{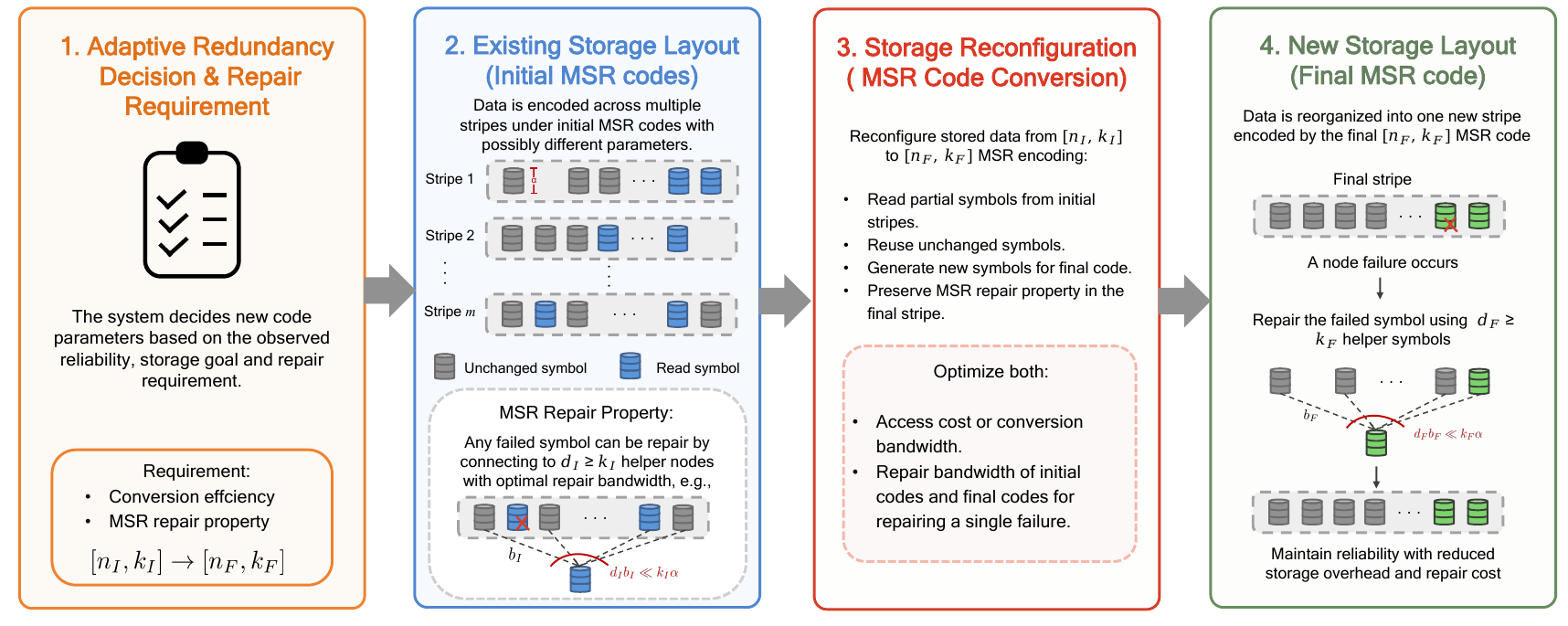}
    \caption{System-level workflow of convertible MSR codes.
    Adaptive redundancy selects new code parameters according
    to reliability, storage, and repair requirements. The stored
    data are then converted from initial MSR codes to a final MSR
    code while preserving efficient node repair after conversion.}
    \label{fig:convertible-msr}
\end{figure*}

\subsection{Motivation: Adaptive Redundancy with Efficient Repair}

Large-scale distributed storage systems encode data across
$n$ storage nodes using an $[n,k]$ erasure code, providing
tolerance against up to $n-k$ simultaneous node failures.
Production deployments at Facebook (HDFS-RAID), Google
(GFS~II), and Microsoft (Windows Azure Storage) rely on such
codes to protect petabyte-scale data. In practice, however,
storage-node reliability is not static: it varies significantly
across device types and over the lifetime of individual devices,
following the well-known bathtub curve of infant mortality,
useful life, and wearout~\cite{kadekodi2019Cluster}. This
motivates \emph{adaptive redundancy}, where the system dynamically
updates the code parameters from $(n_I,k_I)$ to $(n_F,k_F)$
in response to observed reliability, reducing storage overhead
by $11\%$--$44\%$ in production clusters~\cite{kadekodi2019Cluster}.

For the MSR setting considered in this paper, adaptive redundancy
must satisfy two requirements simultaneously, as shown in
Fig.~\ref{fig:convertible-msr}. First, the system should perform
\emph{efficient code conversion}: existing data encoded under
the initial MSR codes should be converted into a final codeword
while accessing only a small subset of the stored symbols.
Second, the resulting final code must still support
\emph{efficient node repair}: after conversion, any failed node
in the final stripe should be repaired by contacting $d$ helper
nodes and downloading the minimum amount of data allowed by the
MSR repair bound.

Existing convertible-code constructions focus on MDS codes and do not preserve the MSR repair property, while standard MSR codes do not support efficient conversion. This work addresses both requirements simultaneously by constructing  \emph{convertible MSR codes}, which combine efficient
MSR-to-MSR conversion with optimal repair of both the initial and
final codes.

\subsection{Code Conversion Model}

We consider the \emph{merge regime},
where $m$ initial codewords encoding disjoint data blocks are
converted into a single final codeword. This corresponds to the
storage reconfiguration stage in Fig.~\ref{fig:convertible-msr}.
Specifically, let
$\mathcal{C}_{I_1},\ldots,\mathcal{C}_{I_m}$ be $m$ initial
MSR array codes over $\mathbb{F}_q$, where
$\mathcal{C}_{I_i}$ is an $(n_{I_i},k_{I_i};\alpha)_q$ code
for each $i\in[m]$. The parameters of the initial codes may be
different, which gives the irregular setting. The conversion
produces a final codeword of an $(n_F,k_F;\alpha)_q$ MSR code
$\mathcal{C}_F$, where
$
    k_F=\sum_{i=1}^{m} k_{I_i}.
$

Fig.~\ref{fig:reconfiguration-b} illustrates this conversion at
the symbol level. The dashed boxes represent initial and final
stripes, and each square represents one vector symbol. Since we
consider array codes, each vector symbol stores $\alpha$
subsymbols, as indicated by the callout in the figure. Selected
read symbols are accessed from the initial stripes and routed to
the reconfiguration engine, which executes the conversion
procedure $\sigma$.

During conversion, symbols are classified into three categories:
\begin{itemize}
    \item \textbf{Unchanged symbols:} symbols that appear
    identically in both an initial codeword and the final
    codeword, and are reused directly in the final codeword;
    \item \textbf{Read symbols:} symbols accessed from the
    initial codewords as inputs to the conversion procedure
    $\sigma$;
    \item \textbf{Written symbols:} symbols of the final codeword
    that are newly computed from the read symbols.
\end{itemize}

The reconfiguration engine in Fig.~\ref{fig:reconfiguration-b}
summarizes the conversion process. It reads the
required symbols from the initial stripes, computes new symbols
using the conversion procedure $\sigma$, and generates the final
stripe by combining unchanged symbols and newly computed symbols.
The objective is to preserve as many unchanged symbols as possible
while minimizing the symbols and subsymbols accessed during
conversion. These requirements are captured by the access-cost
and conversion-bandwidth metrics defined below.

\begin{figure*}[!t]
    \centering
    \resizebox{2\columnwidth}{!}{
    \begin{tikzpicture}[x=1cm,y=1cm]

    \node[font=\bfseries\large] at (8.80,2.20) {};

    \node[font=\bfseries\small] at (3.05,1.62)
    {Initial Stripes};

    \def\sx{0.85}
    \def\sw{4.45}
    \def\sh{0.88}

    \node[font=\small] at (0.45,0.95) {$C_{I_1}$};

    \draw[dashbox]
        (\sx,0.51) rectangle ++(\sw,\sh);

    \node[read] at (1.20,0.95) {};


    \draw[
        dashed,
        draw=redbox,
        line width=0.55pt
    ]
        (1.20,0.95) ellipse [x radius=0.25cm, y radius=0.25cm];

    \draw[
        dashed,
        draw=redbox,
        line width=0.55pt
    ]
        (1.03,1.13) -- (0.58,1.58);

    \draw[
        draw=darkgray,
        line width=0.45pt,
        fill=readblue
    ]
        (0.05,1.58) rectangle (0.95,2.12);

    \draw[darkgray, line width=0.32pt] (0.23,1.58) -- (0.23,2.12);
    \draw[darkgray, line width=0.32pt] (0.41,1.58) -- (0.41,2.12);
    \draw[darkgray, line width=0.32pt] (0.59,1.58) -- (0.59,2.12);
    \draw[darkgray, line width=0.32pt] (0.77,1.58) -- (0.77,2.12);

    \node[
        font=\fontsize{5}{5.5}\selectfont,
        align=center,
        text=darkgray
    ] at (0.50,2.32)
    {store $\alpha$ subsymbols};

    \node[read] at (1.65,0.95) {};
    \node[unchanged] at (2.48,0.95) {};
    \node[read] at (2.92,0.95) {};
    \node[font=\small] at (3.56,0.95) {$\cdots$};
    \node[unchanged] at (4.22,0.95) {};
    \node[unchanged] at (4.70,0.95) {};

    \node[font=\small] at (0.45,-0.30) {$C_{I_2}$};

    \draw[dashbox]
        (\sx,-0.74) rectangle ++(\sw,\sh);

    \node[read] at (1.65,-0.30) {};
    \node[read] at (2.08,-0.30) {};
    \node[unchanged] at (2.48,-0.30) {};
    \node[read] at (2.92,-0.30) {};
    \node[font=\small] at (3.56,-0.30) {$\cdots$};
    \node[read] at (4.22,-0.30) {};

    \node[font=\large] at (2.90,-1.12) {$\vdots$};

    \node[font=\small] at (0.45,-1.92) {$C_{I_m}$};

    \draw[dashbox]
        (\sx,-2.36) rectangle ++(\sw,\sh);

    \node[read] at (1.20,-1.92) {};
    \node[unchanged] at (1.65,-1.92) {};
    \node[read] at (2.08,-1.92) {};
    \node[unchanged] at (2.48,-1.92) {};
    \node[unchanged] at (2.92,-1.92) {};
    \node[font=\small] at (3.56,-1.92) {$\cdots$};
    \node[unchanged] at (4.22,-1.92) {};
    \node[read] at (4.70,-1.92) {};


    \def\busOne{0.66}
    \def\busTwo{-0.60}
    \def\busM{-2.22}

    \coordinate (trunkTop) at (6.20,\busOne);
    \coordinate (trunkMid) at (6.20,\busTwo);
    \coordinate (trunkBot) at (6.20,\busM);
    \coordinate (mergeOut) at (6.88,\busTwo);

    \draw[darkgray!80, line width=0.42pt]
        (1.20,0.78) -- (1.20,\busOne);
    \draw[darkgray!80, line width=0.42pt]
        (1.65,0.78) -- (1.65,\busOne);
    \draw[darkgray!80, line width=0.42pt]
        (2.92,0.78) -- (2.92,\busOne);

    \draw[darkgray!80, line width=0.60pt]
        (1.20,\busOne) -- (2.92,\busOne);

    \draw[darkgray, line width=0.85pt]
        (2.92,\busOne) -- (5.55,\busOne) -- (trunkTop);

    \draw[darkgray!80, line width=0.42pt]
        (1.65,-0.48) -- (1.65,\busTwo);
    \draw[darkgray!80, line width=0.42pt]
        (2.08,-0.48) -- (2.08,\busTwo);
    \draw[darkgray!80, line width=0.42pt]
        (2.92,-0.48) -- (2.92,\busTwo);
    \draw[darkgray!80, line width=0.42pt]
        (4.22,-0.48) -- (4.22,\busTwo);

    \draw[darkgray!80, line width=0.60pt]
        (1.65,\busTwo) -- (4.22,\busTwo);

    \draw[darkgray, line width=0.85pt]
        (4.22,\busTwo) -- (trunkMid);

    \draw[darkgray!80, line width=0.42pt]
        (1.20,-2.10) -- (1.20,\busM);
         \draw[darkgray!80, line width=0.42pt]
        (2.08,-2.10) -- (2.08,\busM);
    \draw[darkgray!80, line width=0.42pt]
        (4.70,-2.10) -- (4.70,\busM);

    \draw[darkgray!80, line width=0.60pt]
        (1.20,\busM) -- (4.70,\busM);

    \draw[darkgray, line width=0.85pt]
        (4.70,\busM) -- (trunkBot);

    \draw[darkgray, line width=0.95pt]
        (trunkTop) -- (trunkBot);

    \draw[darkgray, line width=0.95pt]
        (trunkMid) -- (mergeOut);

    \draw[
        -{Triangle[length=3.0mm,width=3.0mm]},
        line width=1.5pt,
        draw=darkgray
    ]
        (mergeOut) -- (7,\busTwo);

\draw[
    rounded corners=4pt,
    dashed,
    draw=redbox,
    fill=orangebox!7,
    line width=0.8pt
]
    (7.00,-2.35) rectangle (10.55,1.72);

\node[
    font=\bfseries\fontsize{7.2}{7.8}\selectfont,
    text=redbox,
    align=center
] at (8.72,1.44)
{Reconfiguration Engine};


\draw[
    rounded corners=8pt,
    draw=bluebox!60!black,
    fill=bluebox!6,
    line width=0.55pt
]
    (7.12,0.32) rectangle (10.43,1.22);

\draw[darkgray, line width=0.48pt, fill=white]
    (7.52,0.55) ellipse [x radius=0.125, y radius=0.045];
\draw[darkgray, line width=0.48pt]
    (7.395,0.55) -- (7.395,0.90);
\draw[darkgray, line width=0.48pt]
    (7.645,0.55) -- (7.645,0.90);
\draw[darkgray, line width=0.48pt, fill=white]
    (7.52,0.90) ellipse [x radius=0.125, y radius=0.045];

\draw[darkgray, line width=0.52pt, fill=white]
    (7.86,0.80) circle (0.085);
\draw[darkgray, line width=0.52pt]
    (7.92,0.74) -- (8.05,0.58);

\node[
    anchor=west,
    font=\bfseries\fontsize{5.65}{6.15}\selectfont,
    text=darkgray
] at (8,1.02)
{1. Read Symbols};

\node[
    anchor=west,
    font=\fontsize{4.85}{5.35}\selectfont,
    text=darkgray,
    text width=2.24cm,
    align=left
] at (8,0.70)
{Access required symbols from initial stripes.};

\draw[
    rounded corners=8pt,
    draw=greenbox!60!black,
    fill=greenbox!7,
    line width=0.55pt
]
    (7.12,-0.82) rectangle (10.43,0.08);

\draw[darkgray, line width=0.50pt, fill=white]
    (7.39,-0.55) rectangle (7.82,-0.17);
\draw[darkgray, line width=0.42pt, fill=gray!8]
    (7.48,-0.455) rectangle (7.72,-0.265);

\foreach \y in {-0.49,-0.41,-0.33,-0.25}{
    \draw[darkgray, line width=0.36pt] (7.32,\y) -- (7.39,\y);
    \draw[darkgray, line width=0.36pt] (7.82,\y) -- (7.89,\y);
}
\foreach \x in {7.47,7.56,7.65,7.74}{
    \draw[darkgray, line width=0.36pt] (\x,-0.63) -- (\x,-0.55);
    \draw[darkgray, line width=0.36pt] (\x,-0.17) -- (\x,-0.09);
}

\node[
    anchor=west,
    font=\bfseries\fontsize{5.55}{6.05}\selectfont,
    text=darkgray
] at (8,-0.13)
{2. Compute new symbols};

\node[
    anchor=west,
    font=\fontsize{4.35}{4.90}\selectfont,
    text=darkgray,
    text width=2.28cm,
    align=left
] at (8,-0.47)
{Use read symbols and compute new symbols by conversion procedure $\sigma$.};

\draw[
    rounded corners=8pt,
    draw=purple!60!black,
    fill=purple!7,
    line width=0.55pt
]
    (7.12,-1.96) rectangle (10.43,-1.06);

\node[
    rectangle,
    draw=darkgray,
    fill=white,
    minimum width=0.16cm,
    minimum height=0.16cm,
    inner sep=0pt,
    line width=0.40pt
] at (7.38,-1.52) {};

\node[
    rectangle,
    draw=darkgray,
    fill=white,
    minimum width=0.16cm,
    minimum height=0.16cm,
    inner sep=0pt,
    line width=0.40pt
] at (7.62,-1.52) {};

\node[
    rectangle,
    draw=darkgray,
    fill=white,
    minimum width=0.16cm,
    minimum height=0.16cm,
    inner sep=0pt,
    line width=0.40pt,
    dashed
] at (7.86,-1.52) {};

\node[
    rectangle,
    draw=darkgray,
    fill=white,
    minimum width=0.16cm,
    minimum height=0.16cm,
    inner sep=0pt,
    line width=0.40pt,
    dashed
] at (8.10,-1.52) {};

\node[
    anchor=west,
    font=\bfseries\fontsize{5.55}{6.05}\selectfont,
    text=darkgray
] at (8,-1.30)
{3. Generate Final Stripe};

\node[
    anchor=west,
    font=\fontsize{4.25}{4.80}\selectfont,
    text=darkgray,
    text width=2.18cm,
    align=left
] at (8.1,-1.66)
{Combine unchanged symbols and newly computed symbols into the target final stripe.};


\draw[darkgray, line width=0.50pt]
    (8.78,0.28) -- (8.78,0.20);
\path[fill=darkgray]
    (8.78,0.10) -- (8.705,0.23) -- (8.855,0.23) -- cycle;

\draw[darkgray, line width=0.50pt]
    (8.78,-0.86) -- (8.78,-0.94);
\path[fill=darkgray]
    (8.78,-1.04) -- (8.705,-0.91) -- (8.855,-0.91) -- cycle;

    \draw[midarrow] (10.55,0.55) -- (11.35,0.55);

    \node[font=\bfseries\small] at (13.90,1.30)
    {Final Stripe $C_F$};

    \draw[dashbox]
        (11.35,0.08) rectangle (16.98,1.02);

    \node[unchanged] at (11.85,0.55) {};
    \node[unchanged] at (12.45,0.55) {};
    \node[unchanged] at (13.05,0.55) {};
    \node[unchanged] at (13.65,0.55) {};
    \node[font=\small] at (14.25,0.55) {$\cdots$};
    \node[unchanged] at (14.85,0.55) {};
    \node[newsym] at (15.45,0.55) {};
    \node[newsym] at (16.05,0.55) {};
    \node[newsym] at (16.65,0.55) {};

    \draw[
        rounded corners=3pt,
        draw=gray!60,
        fill=white,
        line width=0.45pt
    ]
        (11.25,-2.58) rectangle (17.55,-0.72);

    \node[read] at (11.70,-1.12) {};
    \node[
        anchor=west,
        font=\fontsize{5.8}{6.5}\selectfont,
        text width=5.20cm,
        align=left
    ] at (12.15,-1.12)
    {Read symbols: accessed from initial stripes};

    \node[unchanged] at (11.70,-1.60) {};
    \node[
        anchor=west,
        font=\fontsize{5.8}{6.5}\selectfont,
        text width=5.20cm,
        align=left
    ] at (12.15,-1.60)
    {Unchanged symbols: reused directly};

    \node[newsym] at (11.70,-2.08) {};
    \node[
        anchor=west,
        font=\fontsize{5.8}{6.5}\selectfont,
        text width=5.20cm,
        align=left
    ] at (12.15,-2.08)
    {New symbols: generated for the final stripe};

    \end{tikzpicture}
    }
    \caption{Symbol-level illustration of code conversion in the merge regime.
    Irregular settings may involve different initial codes,
    but the stripe lengths are drawn uniformly to avoid implying different parameters.}
    \label{fig:reconfiguration-b}
\end{figure*}
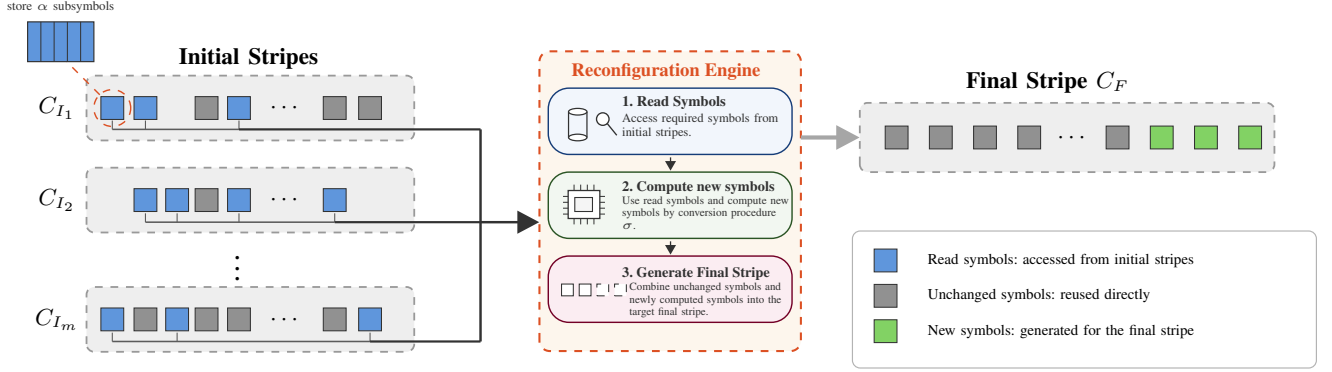

\subsection{Formal Definition and Performance Metrics}

\begin{defn}[Irregular convertible code in the merge
regime~\cite{ge2024MDS}]
Let $m$ be a positive integer. An $(m,1)_q$ irregular
convertible code $\mathscr{C}$ over $\mathbb{F}_q$ consists of
\begin{enumerate}
\item[1)] $m$ initial codes
$\mathcal{C}_{I_1},\ldots,\mathcal{C}_{I_m}$ and a final code
$\mathcal{C}_F$, where $\mathcal{C}_{I_i}$ is an
$(n_{I_i},k_{I_i};\alpha)_q$ code for $i\in[m]$, and
$\mathcal{C}_F$ is an $(n_F,k_F;\alpha)_q$ code satisfying
$\sum_{i=1}^m k_{I_i}=k_F$;
\item[2)] an injective conversion procedure
$
\sigma:\prod_{i=1}^m \mathcal{C}_{I_i}\rightarrow
\mathcal{C}_F.
$
\end{enumerate}
\end{defn}

The efficiency of conversion is quantified by two metrics.
The \emph{access cost} $\rho=\rho_R+\rho_W$ counts the total
number of symbols touched by the conversion, where $\rho_R$ is
the number of read symbols and $\rho_W$ is the number of newly
written final-code symbols. The \emph{conversion bandwidth}
$\gamma=\gamma_R+\gamma_W$ measures the total amount of data
transferred at the subsymbol level, where $\gamma_R$ and
$\gamma_W$ are the read and write bandwidths, respectively.
Since each vector symbol consists of $\alpha$ subsymbols,
conversion bandwidth captures a finer-grained cost than symbol
access. These metrics serve as the design targets for the
constructions in this paper, guided by the following lower
bounds.

\begin{thm}[\cite{ge2024MDS}]\label{bound on access}
For any $(m,1)_q$ MDS irregular convertible code, the access
cost satisfies
\begin{equation*}
\rho \ge r_{F} + \sum_{\substack{i\in[m],\,r_{F}\le\\
\min\{k_{I_i},r_{I_i}\}}} r_{F}
+ \sum_{\substack{i\in[m],\,r_{F}>\\
\min\{k_{I_i},r_{I_i}\}}} k_{I_i},
\end{equation*}
where $r_{F}=n_{F}-k_{F}$ and $r_{I_i}=n_{I_i}-k_{I_i}$.
\end{thm}

\begin{thm}[\cite{maturana2023Bandwidth}]\label{bound on bandwidth}
For any $(m,1)_q$ MDS convertible code with $m$ identical
initial $(n_I,k_I;\alpha)_q$ codes and one final
$(n_F,mk_I;\alpha)_q$ code, the read conversion bandwidth
satisfies
\begin{equation*}
\gamma_R \ge \begin{cases}
m\alpha\min\{k_I,r_F\}, & r_I\ge r_F \text{ or } k_I\le r_F,\\
m\alpha\!\left(r_I+k_I\!\left(1-\tfrac{r_I}{r_F}\right)\right),
& \text{otherwise},
\end{cases}
\end{equation*}
and $\gamma_W\ge r_F\alpha$, where $r_I=n_I-k_I$.
\end{thm}

The second requirement in Fig.~\ref{fig:convertible-msr} is that
efficient repair must remain available after conversion. In other
words, the conversion should not merely produce an MDS final code;
it should produce a final code that is still an MSR code. This
ensures that once the system transitions to the new code
configuration, individual node failures can still be repaired with
optimal repair bandwidth. This leads to the following definition.

\begin{defn}[Convertible MSR code]
An $(m,1)_q$ irregular convertible code $\mathscr{C}$ is a
\emph{convertible MSR code} if each initial code
$\mathcal{C}_{I_i}$, $i\in[m]$, and the final code
$\mathcal{C}_F$ are all MSR codes.
\end{defn}

\section{A General Conversion Framework for MDS Array Codes}\label{sec: framework}

Motivated by the scalar MDS convertible code in~\cite{ge2024MDS}, in this section we present a general conversion framework from $m$ $(n_{I_i},k_{I_i};\alpha)_q$ MDS array codes, $i\in[m]$, to an $(n_F,k_F;\alpha)_q$ MDS array code.

Throughout this section, we focus on MDS array codes. Specifically, for each $i\in[m]$ and each subsymbol index $w\in[0,\alpha-1]$, the coordinate-wise vector $(c^{(i)}_{w,0},\cdots,c^{(i)}_{w,n_{I_i}-1})$ forms an $(n_{I_i},k_{I_i};1)_q$ scalar MDS codeword. Therefore, the conversion can be carried out row-wise for each fixed $w$, and the final array code is obtained by stacking all rows.

We first describe the initial-code and final-code setup, and then specify the conversion procedures $\sigma_1$ and $\sigma_2$ under different parameter regimes.

Define $n_F = r_F+\sum_{i=1}^{m} k_{I_i}$ and $k_F=\sum_{i=1}^{m} k_{I_i}$. An $(m,1)_q$ convertible code $\mathscr{C}$ consists of the following three parts:

\textbf{Initial codes:} 
For $i\in[m]$, choose the initial $(n_{I_i},k_{I_i};\alpha)_q$ MDS array code $\mathcal{C}_{I_i}$. Let $\mathbf{c}_i=(\mathbf{c}_{i,0},\ldots,\mathbf{c}_{i,n_{I_i}-1})$ be a codeword of $\mathcal{C}_{I_i}$. 

For each $w\in[0,\alpha-1]$, the row $(c^{(i)}_{w,0},\cdots,c^{(i)}_{w,n_{I_i}-1})$ is a codeword of an $(n_{I_i},k_{I_i};1)_q$ MDS code with an $r_{I_i}\times n_{I_i}$ parity-check matrix
$H_{i,w}.$

\textbf{Final code:} 
The final code consists of the first $\sum_{i=1}^m k_{I_i}$ symbols copied from the initial codes and $r_F$ newly generated symbols.

Define the final code as
\begin{equation*}
\mathcal{C}_F := \{ \sigma(\mathbf{c}_1,\dots,\mathbf{c}_{m}) : \mathbf{c}_i\in\mathcal{C}_{I_i}, i\in[m]  \},
\end{equation*}
where $\sigma$ is a conversion procedure defined as follows.

\textbf{Conversion procedures: } 
For each fixed row $w$, the conversion follows three steps: 
(i) identify the unchanged symbols and read symbols, 
(ii) compute the written symbols, and
(iii) deduce the parity-check matrix of the final code. 

We consider the following two parameter regimes:
\begin{itemize}
\item Regime I: $r_F\le \min\{k_{I_i},r_{I_i}\}$ for all $i\in[m]$, where we use conversion procedure $\sigma_1$;
\item Regime II: $r_F>\min\{k_{I_i},r_{I_i}\}$ for all $i\in[m]$, where we use conversion procedure $\sigma_2$.
\end{itemize}

We describe the two procedures separately.

\begin{itemize}

\item[*]$\sigma_1:$ 
\textbf{Step (i):}
For each $i\in[m]$, we define:
1) the \emph{unchanged symbols} as the first $k_{I_i}$ symbols, 
i.e., $\{\mathbf{c}_{i,j}: j\in[0,k_{I_i}-1]\}$;
2) the \emph{read symbols} as the symbols indexed from $k_{I_i}$ 
to $k_{I_i}+r_F-1$, i.e., $\{\mathbf{c}_{i,j}: j\in[k_{I_i},
k_{I_i}+r_F-1]\}$;
3) the \emph{written symbols} as $\hat{\mathbf{c}}=(\hat{\mathbf{c}}_0,
\hat{\mathbf{c}}_1,\cdots,\hat{\mathbf{c}}_{r_F-1})$, where 
$\hat{\mathbf{c}}_j=(\hat{c}_{0,j},\cdots,\hat{c}_{\alpha-1,j})$, 
$j\in[0,r_F-1]$.

Thus, the final codeword can be represented by
\begin{equation}\label{eq: final_codeword_1}
\mathbf{c}_F = (\mathbf{c}_1|_{[0,k_{I_1}-1]}, \dots, 
\mathbf{c}_m|_{[0,k_{I_{m}}-1]}, \hat{\mathbf{c}}).
\end{equation} 

\textbf{Step (ii):}
For each $i\in[m]$, let $\mathcal{P}_i=\{0,1,\cdots,k_{I_i}+r_F-1\}$ 
denote the set consisting of the first $k_{I_i}$ symbols and an 
additional $r_F$ symbols. Since $\mathcal{C}_{I_i}$ is an MDS code, 
its punctured code $\mathcal{C}_{I_i}|_{\mathcal{P}_i}$ is also MDS, 
and the parity-check matrix of $\mathcal{C}_{I_i}|_{\mathcal{P}_i}$ 
can be written as
\begin{equation}\label{eq: punctured_pc}
H_{i,w}|_{\mathcal{P}_i}= [\hat{H}_{i,w}^{(\mathcal{P}_i)} \mid 
\bar{H}_{i,w}^{(\mathcal{P}_i)}], \quad w\in[0,\alpha-1],
\end{equation}
where $\hat{H}_{i,w}^{(\mathcal{P}_i)}$ is an $r_{F}\times k_{I_i}$ 
matrix corresponding to the first $k_{I_i}$ symbols, and 
$\bar{H}_{i,w}^{(\mathcal{P}_i)}$ is an $r_{F}\times r_{F}$ 
invertible matrix corresponding to the remaining $r_F$ symbols. 
This gives
\begin{equation}\label{eq: pc_ic}
\begin{split}
\hat{H}_{i,w}^{(\mathcal{P}_i)}({c}^{(i)}_{w,0},&\cdots,
{c}^{(i)}_{w,k_{I_i}-1})^{\top}=\\
&-\bar{H}_{i,w}^{(\mathcal{P}_i)}({c}^{(i)}_{w,k_{I_i}},\cdots,
{c}^{(i)}_{w,k_{I_i}+r_F-1})^{\top}.
\end{split}
\end{equation}
Given an $r_F\times r_F$ invertible matrix $\hat{H}_w$ to be 
specified in each construction, the written symbols are computed as
\begin{equation}\label{eqn conA}
\begin{split}
(&\hat{{c}}_{w,0},\hat{{c}}_{w,1},\cdots,\hat{{c}}_{w,r_F-1})^\top= \\
&\hat{H}^{-1}_w \sum_{i=1}^{m} \bar{H}_{i,w}^{(\mathcal{P}_i)} 
\cdot (c^{(i)}_{w,k_{I_i}},\cdots,c^{(i)}_{w,k_{I_i}+r_F-1})^\top.
\end{split}
\end{equation}

\textbf{Step (iii):}
It follows from~\eqref{eq: final_codeword_1}, \eqref{eq: pc_ic}, 
and~\eqref{eqn conA} that the final code $\mathcal{C}_F$ has 
parity-check matrix
\begin{equation}\label{eq: parity_sigma_1}
H_{F,w}:=
[\hat{H}_{1,w}^{(\mathcal{P}_1)}| \cdots| \hat{H}_{m,w}^{(\mathcal{P}_m)}
| \hat{H}_w], \,w\in[0,\alpha-1].
\end{equation}
The MDS property of $\mathcal{C}_F$ is determined by whether 
$H_{F,w}$ is the parity-check matrix of an $(n_F,k_F;1)_q$ MDS 
code for each $w$.
\item[*]$\sigma_2:$ 
\textbf{Step (i):}
1) For each $i\in[m]$, the \emph{unchanged symbols} and 
\emph{read symbols} both consist of the first $k_{I_i}$ symbols, 
i.e., $\{\mathbf{c}_{i,j}: j\in[0,k_{I_i}-1]\}$;
2) The \emph{written symbols} are defined as 
$\hat{\mathbf{c}}=(\hat{\mathbf{c}}_0,\hat{\mathbf{c}}_1,\cdots,
\hat{\mathbf{c}}_{r_F-1})$, where $\hat{\mathbf{c}}_j=(\hat{c}_{0,j},
\cdots,\hat{c}_{\alpha-1,j})$, $j\in[0,r_F-1]$.

The final codeword is still given by~\eqref{eq: final_codeword_1}.

\textbf{Step (ii):}
For each $w\in[0,\alpha-1]$, write the parity-check matrix of 
$\mathcal{C}_{I_i}$ as $H_{i,w}=[\hat{H}_{i,w}\mid\bar{H}_{i,w}]$, 
where $\hat{H}_{i,w}$ is an $r_{I_i}\times k_{I_i}$ matrix and 
$\bar{H}_{i,w}$ is an $r_{I_i}\times r_{I_i}$ invertible matrix. 
Let $\hat{H}_{i,w}^E$ be an $r_F\times k_{I_i}$ matrix obtained 
from $\hat{H}_{i,w}$ by appending $r_F-r_{I_i}$ additional rows. 

The written symbols are computed as
\begin{equation}\label{eqn conA_rF_gt}\small
\begin{split}
(\hat{{c}}_{w,0},\hat{{c}}_{w,1},&\cdots,\hat{{c}}_{w,r_F-1})^\top =\\
&- \hat{H}_w^{-1} \sum_{i=1}^{m} \hat{H}_{i,w}^E \cdot 
({c}^{(i)}_{w,0},\cdots,{c}^{(i)}_{w,k_{I_i}-1})^{\top},
\end{split}
\end{equation}
where $\hat{H}_w$ is an $r_F\times r_F$ invertible matrix.

\textbf{Step (iii):}
With the written symbols computed as in~\eqref{eqn conA_rF_gt} and the final codeword \eqref{eq: final_codeword_1}, 
the final code $\mathcal{C}_F$ has parity-check matrix
\begin{equation*}
H_{F,w}:=
[\hat{H}_{1,w}^{E}| \cdots| \hat{H}_{m,w}^{E}| \hat{H}_w],
\quad w\in[0,\alpha-1].
\end{equation*}
The MDS property of $\mathcal{C}_F$ is determined by whether 
$H_{F,w}$ is the parity-check matrix of an $(n_F,k_F;1)_q$ MDS 
code for each $w$, as specified in Lemma~\ref{lem: mds_condition}.

\end{itemize}

The above framework enables row-wise conversion, where each row is treated as a scalar MDS code and stacked to form the array code.

The next lemma follows directly from Theorem~4 in~\cite{ge2024MDS} via row-wise application. Specifically, for each fixed subsymbol index $w$, the setting is scalar, and thus the corresponding scalar MDS and scalar access-optimal results apply. Hence, if each row code is MDS, then stacking all rows yields the array-code MDS property. Similarly, if each row conversion is scalar access-optimal, then summing over all rows yields array-level access optimality, since access cost is measured in symbols. Therefore, the proofs are omitted for brevity.

\begin{lem}\label{lem: mds_condition}
By the above conversion framework, for each $w\in[0,\alpha-1]$, denote the parity-check matrix of the final code by
\[
H_{F,w}:=
\begin{cases}
[\hat{H}_{1,w}^{(\mathcal{P}_1)}| \cdots| \hat{H}_{m,w}^{(\mathcal{P}_m)}| \hat{H}_w], & \text{for}\, \sigma_1,\\
[\hat{H}^E_{1,w}| \cdots| \hat{H}^E_{m,w}| \hat{H}_w], & \text{for} \,\sigma_2.
\end{cases}
\]

If $H_{F,w}$ is the parity-check matrix of an $(n_F,k_F;1)_q$ MDS code, then the above array code $\mathscr{C}$ is an MDS convertible code.
Furthermore,  $\mathscr{C}$ achieves optimal access cost.
\end{lem}

\begin{remark}
When $\alpha=1$, the above conversion framework reduces to the scalar setting and is consistent with~\cite{ge2024MDS}.
\end{remark}

\section{Array MSR-to-MSR Conversion}\label{sec: array}

In this section, we consider $(m,1)_q$ convertible MSR array 
codes. We first address the irregular setting, where the initial 
codes may have different parameters. Building on this, we further consider the practically 
motivated same-code setting, where all initial codewords are 
drawn from the same MSR code.

\subsection{Irregular Constructions}

We begin with the irregular setting, where the $m$ initial codes 
$\mathcal{C}_{I_i}$ for $i\in[m]$ are Hadamard-based MSR codes, which may have different parameters $(n_{I_i},k_{I_i};\alpha)$. Building on the conversion framework established in Section~\ref{sec: framework}, we construct $(m,1)_q$ convertible MSR codes achieving optimal access cost.

The main challenge arises from a structural mismatch between MSR repair and code conversion. In a Hadamard MSR code, the coefficient 
at node $u$ depends on $w$ only through the $u$-th coordinate 
$w_{(u)}$, which is the key property enabling interference 
alignment during repair. However, when multiple Hadamard MSR codes 
are merged directly via conversion, the resulting final code does 
not inherit this property: different nodes in the final code end 
up sharing the same coordinate dependence, violating the 
coordinate-wise periodicity required for MSR repair. 

To address this issue, we design the initial codes with a subsymbol-level structure such that the required periodicity is preserved under conversion. More precisely, we observe that the periodicity required by the final code can be embedded into the initial codes at the design stage, by arranging their subsymbol-level structure according to the parameters of the final 
code. As we will show, this pre-alignment does not compromise the MSR repair properties of the initial codes, and ensures that the required periodic structure is automatically inherited by the final code after conversion. An example is given below.

\begin{ex}\label{ex: motivating}
We use the parameters $m=2$, $n_I=4$, $k_I=2$, $d_I=3$, $d_F=5$, 
$r_F=2$. Consider two $(4,2;2^4)_q$ Hadamard 
MSR codes $\mathcal{C}_{I_1}$ and $\mathcal{C}_{I_2}$ with 
$s_I=d_I-k_I+1=2$. By~\eqref{eq: had_pc_row} 
and~\eqref{eq: had_row_points_simple}, the parity-check equations 
of $\mathcal{C}_{I_i}$ are given by 
\begin{equation*}
\sum_{u=0}^{3}(\lambda^{(i)}_{w_{(u)},\,u})^{\,t}\,c^{(i)}_{w,u}=0,
\quad t\in[0,1],\;w\in[0,15],\;i\in[2].
\end{equation*}

If the two initial codes are merged directly using the conversion 
framework, the coordinate-wise dependence required for the Hadamard 
structure is not preserved. More precisely, by our conversion 
framework, the parity-check equations of the final code are
\begin{equation*}\small
\begin{split}
&\sum_{u=0}^{1}(\lambda^{(1)}_{w_{(u)},\,u})^{\,t}\,c_{w,u}
+\sum_{u=2}^{3}(\lambda^{(2)}_{w_{(u-2)},\,u-2})^{\,t}\,c_{w,u}\\
&+\sum_{u=4}^{5}\hat{\lambda}_{w_{(u)},\,u}^{\,t}\,c_{w,u}=0,
\quad t\in[0,1],\;w\in[0,15].
\end{split}
\end{equation*}
Recall that in a Hadamard MSR code, the coefficient at node $u$ 
depends on $w$ only through the $u$-th coordinate $w_{(u)}$. 
However, after merging, the coefficients at nodes $0$ and $2$ 
both depend on $w_{(0)}$, and similarly nodes $1$ and $3$ both 
depend only on $w_{(1)}$. Hence different nodes share the same 
periodic pattern, violating the coordinate-wise periodicity 
required for MSR repair.

To resolve this, we increase the subpacketization to $\alpha=2^6=64$ 
and pre-align the initial codes according to the parameters of the 
final code. Specifically, $\mathcal{C}_{I_1}$ retains its row 
coefficient structure with the extended subpacketization:
\begin{equation*}\small
\sum_{u=0}^{3}(\lambda^{(1)}_{w_{(u)},\,u})^{\,t}\,c^{(1)}_{w,u}=0,
\quad t\in[0,1],\;w\in[0,63],
\end{equation*}
while $\mathcal{C}_{I_2}$ is additionally re-indexed by replacing 
$w_{(u)}$ with $w_{(u+2)}$ in its row coefficients:
\begin{equation*}
\sum_{u=0}^{3}(\lambda^{(2)}_{w_{(u+2)},\,u})^{\,t}\,c^{(2)}_{w,u}=0,
\quad t\in[0,1],\;w\in[0,63].
\end{equation*} As we will show in the proof of Theorem~\ref{thm: msr}, this pre-alignment does not compromise 
the MSR repair properties of the initial codes. 

Then, define
\begin{equation*}
\begin{split}
&\lambda_{w_{(u)},u}:=\lambda^{(1)}_{w_{(u)},\,u},\quad u=0,1,\\
&\lambda_{w_{(u)},u}:=\lambda^{(2)}_{w_{(u)},\,u-2},\quad u=2,3,\\
&\lambda_{w_{(u)},u}:=\hat{\lambda}_{w_{(u)},\,u},\quad u=4,5,
\end{split}
\end{equation*}
where $\{\hat{\lambda}_{v,u}: v\in\{0,1\},\, u\in\{4,5\}\}$ are 
chosen to be distinct from $\{\lambda^{(1)}_{v,u}: v\in\{0,1\},\,
u\in\{0,1\}\}$ and $\{\lambda^{(2)}_{v,u}: v\in\{0,1\},\,
u\in\{0,1\}\}$, ensuring the MDS property of the final code. After applying the 
conversion framework with 
$\hat{H}_w=\mathrm{Vand}_{2}(\hat{\lambda}_{w_{(u)},u}:u\in[4,5])$, 
the final code has parity-check equations
\begin{equation*}
\sum_{u=0}^{5}\lambda_{w_{(u)},\,u}^{\,t}\,c_{w,u}=0,
\quad t\in[0,1],\;w\in[0,63],
\end{equation*}
where each coefficient depends on $w$ only through its $u$-th 
coordinate in the $2$-ary expansion, which is precisely the 
Hadamard MSR structure in~\eqref{eq: had_pc_row} 
and~\eqref{eq: had_row_points_simple}. 
\end{ex}

We now introduce the notation used in the construction. For $i\in[m]$, let $\mathcal{P}_i:=\{0,1,\cdots,k_{I_i}+r_F-1\}$. For any $w\in[0,s^L-1]$, recall that $w_{(u)}$ denotes the $u$-th coordinate in the $s$-ary expansion of $w$. Let $k_{I_i}<d_{I_i}\leq n_{I_i}-1$ and $k_F<d_F\leq n_F-1$ denote the numbers of helper nodes used for single-node repair in $\mathcal{C}_{I_i}$ and 
$\mathcal{C}_F$, respectively. In both codes, any single-node failure can be repaired with optimal repair bandwidth. According to the lower bound on access cost given in Theorem~\ref{bound on access}, we consider two parameter regimes: $r_F \leq \min\{k_{I_i}, r_{I_i}\}$ and $r_F > \min\{k_{I_i}, r_{I_i}\}$ for all $i\in[m]$.

In the following constructions, we only specify the initial codes, 
the conversion procedure, and the matrix $\hat{H}_w$ for the final code; 
the remaining components follow directly from the framework in 
Section~\ref{sec: framework}.
\begin{cons}\label{cons: array_rI>r_F}
Consider the regime $r_{F}\le\min\{k_{I_i},r_{I_i}\}$ for all $i\in[m]$.
\begin{itemize}
    \item {Parameters.} Define 
    \[\small s:=\mathrm{lcm}(d_{I_1}-k_{I_1}+1,\cdots,d_{I_m}-k_{I_m}+1,d_F-k_F+1)\]
    and $\alpha:=s^L,\, L=\max\{\sum_{i=1}^{j} k_{I_i}+r_{j}:j\in[m]\}.$
    Assume that $q> sL $.

    \item {Initial codes.} Let $\beta$ be a primitive element of $\fq$.
Define
\[\lambda_{v,u}:=\beta^{su+v},\quad v\in[0,s-1],\,u\in[0,L-1].
\]
    
    For any $i\in[m]$, the initial code 
    $\mathcal{C}_{I_i}$ is defined by the parity-check equations
    \begin{equation}\label{eq: pc_initial}\small
    \begin{split}
   \sum_{u=0}^{n_{I_i}-1}\lambda_{w_{(u+\sum_{j=1}^{i-1}k_{I_j})},\,u+\sum_{j=1}^{i-1}k_{I_j}}^{\,t}
\,c^{(i)}_{w,u}=& 0,\\ t \in [0,\,r_{I_i}-1],\,\,&
    w \in [0,\alpha-1].
    \end{split}
    \end{equation}

    \item {Conversion.} We apply conversion procedure $\sigma_1$, with
    \begin{equation}\label{eq: cons1_h}\small
    \hat{H}_w=\mathrm{Vand}_{r_F}(\lambda_{w_{(u)},\,u}:\,u\in[\sum_{i=1}^mk_{I_i},\sum_{i=1}^mk_{I_i}+r_F-1]).
    \end{equation}
\end{itemize}
\end{cons}

\begin{thm}\label{thm: msr}
Construction~\ref{cons: array_rI>r_F} yields an $(m,1)_q$ convertible MSR code with optimal access cost. For each $i\in[m]$, the initial code 
$\mathcal{C}_{I_i}$ is an $(n_{I_i},k_{I_i};\alpha)_q$ MSR code, and the final code $\mathcal{C}_F$ is an $(n_F,k_F;\alpha)_q$ MSR code, where $k_F=\sum_{i=1}^m k_{I_i}$.
\end{thm}

\begin{IEEEproof}
We prove that Construction~\ref{cons: array_rI>r_F} yields a 
convertible MSR code with optimal access cost by verifying the 
MSR property of both the initial codes and the final code.

\textbf{Initial codes are MSR.} 
For each $i\in[m]$, we show that 
the parity-check equations~\eqref{eq: pc_initial} define a 
Hadamard-based MSR code with parameters $(n_{I_i},k_{I_i};\alpha)_q$ 
by verifying the two structural conditions in 
Section~\ref{sec: hadamard}.

For the layered MDS property~\eqref{eq: had_pc_row}, the 
coefficients
\[
\lambda_{w_{(u+\sum_{j=1}^{i-1}k_{I_j})},\,u+\sum_{j=1}^{i-1}k_{I_j}},
\quad u\in[0,n_{I_i}-1]
\]
are distinct for each $w\in[0,\alpha-1]$, so each row forms an 
$(n_{I_i},k_{I_i};1)_q$ MDS code.
For the coordinate-wise periodicity~\eqref{eq: had_row_points_simple}, 
the row coefficient at position $u$ depends only on the 
$(u+\sum_{j=1}^{i-1}k_{I_j})$-th coordinate of $w$ in the $s$-ary 
expansion, which is precisely the structure required for 
interference alignment in MSR repair.

Since $s_{I_i}:=d_{I_i}-k_{I_i}+1$ divides $s$ by definition, 
we partition $[0,s-1]$ into $s/s_{I_i}$ disjoint intervals 
$[as_{I_i},\,(a+1)s_{I_i}-1]$ for $a\in[0,s/s_{I_i}-1]$. For 
each row index $w$, failed node $j^*\in[0,n_{I_i}-1]$, and 
interval $a$, define the sub-group
\[
\mathcal{G}^{(i)}_a(w,j^*):=\bigl\{R_w\bigl(j^*+{\textstyle
\sum_{l=1}^{i-1}}k_{I_l},\,as_{I_i}+b\bigr):b\in[0,s_{I_i}-1]\bigr\}.
\]
The sub-groups $\{\mathcal{G}^{(i)}_a(w,j^*):a\in[0,s/s_{I_i}-1],\,
w\in[0,\alpha-1]\}$ form a partition of $[0,\alpha-1]$. Within 
each sub-group, the coordinate-wise 
periodicity ensures that the 
coefficients at all helper nodes are invariant, and the repair 
procedure follows the same argument as in 
Section~\ref{sec: hadamard}(i) with $s_{I_i}$ in place of $s$ 
and $d_{I_i}$ in place of $d$. Running over all sub-groups gives 
repair bandwidth
$
b=\frac{\alpha}{s_{I_i}}=\frac{\alpha}{d_{I_i}-k_{I_i}+1},
$
meeting the cut-set bound in~\eqref{eq: beta}. Hence each 
$\mathcal{C}_{I_i}$ is an $(n_{I_i},k_{I_i};\alpha)_q$ MSR code 
with optimal repair bandwidth.

\textbf{Final code is MSR.} 
By \eqref{eq: parity_sigma_1} and \eqref{eq: cons1_h}, for each $w\in[0,\alpha-1]$, the 
parity-check matrix $H_{F,w}$ is a Vandermonde matrix whose evaluation points are distinct by construction. Hence $H_{F,w}$ 
is the parity-check matrix of an $(n_F,k_F;1)_q$ MDS code, and stacking over all rows yields an $(n_F,k_F;\alpha)_q$ MDS array code, establishing the layered MDS property~\eqref{eq: had_pc_row} 
for $\mathcal{C}_F$.

It remains to show that $\mathcal{C}_F$ satisfies the coordinate-wise periodicity required for MSR repair. By our conversion framework, the first $k_{I_i}$ nodes of $\mathcal{C}_{I_i}$ are unchanged. Therefore,
for $w\in[0,\alpha-1]$, we can define the final codeword as
\begin{equation*}\small
\begin{split}
&c_{w,\,j+\sum_{l=1}^{i-1}k_{I_l}}:=c^{(i)}_{w,j},
\quad i\in[m],\,j\in[0,k_{I_i}-1],\\
&c_{w,\,j+\sum_{i=1}^mk_{I_i}}:=\hat{c}_{w,j},
\quad j\in[0,r_F-1].
\end{split}
\end{equation*}
By the pre-alignment of the initial codes and the choice of 
$\hat{H}_w$ in~\eqref{eq: cons1_h}, the row coefficients of 
$\mathcal{C}_F$ satisfy
$
\alpha_{w,u} := \lambda_{w_{(u)},\,u}, \quad u\in[0,n_F-1].
$
Then the parity-check equations of $\mathcal{C}_F$ take the form
\begin{equation*}\small
\sum_{u=0}^{n_F-1}\alpha_{w,u}^{\,t}\,c_{w,u}=0,
\quad t\in[0,r_F-1],\;w\in[0,\alpha-1].
\end{equation*}
This is exactly the Hadamard MSR structure 
in~\eqref{eq: had_row_points_simple}, where each coefficient 
$\alpha_{w,u}=\lambda_{w_{(u)},u}$ depends on $w$ only through 
its $u$-th coordinate in the $s$-ary expansion. By the same 
argument as for the initial codes above, with $s_F:=d_F-k_F+1$ 
in place of $s_{I_i}$, $d_F$ in place of $d_{I_i}$, $\mathcal{C}_F$ is an $(n_F,k_F;\alpha)_q$ 
MSR code with $d_F$ helper nodes.

Therefore, Construction~\ref{cons: array_rI>r_F} yields an 
$(m,1)_q$ convertible MSR code, and optimal access cost follows 
from Lemma~\ref{lem: mds_condition}.
\end{IEEEproof}

	\begin{remark}[Construction 1A]\label{rmk: construction2}
For the regime
$
r_F>\min\{k_{I_i},r_{I_i}\},\, i\in[m],
$
the construction is identical to
Construction~\ref{cons: array_rI>r_F},
except for the sub-packetization level and the conversion
procedure. We omit the full statement for brevity.

Define
\begin{equation}\label{eq: sub-packetization_cons2}\small
\alpha=s^L,
\qquad
L=\sum_{i=1}^mk_{I_i}+r_F,
\end{equation}
and assume $q>sL$.
The initial codes are still given by
\eqref{eq: pc_initial}, with $\alpha$
defined by~\eqref{eq: sub-packetization_cons2}.

The final code is generated using conversion procedure
$\sigma_2$, where
\begin{equation}\label{eq: cons2_h}\small
\hat{H}_w
=
\mathrm{Vand}_{r_F}
(\lambda_{w_{(u)},u}:
u\in[\sum_{i=1}^mk_{I_i},
\sum_{i=1}^mk_{I_i}+r_F-1]).
\end{equation}

By the same argument as in
Theorem~\ref{thm: msr},
this yields an $(m,1)_q$ convertible MSR code
with optimal access cost.
\end{remark}

\subsection{Same-Code Constructions}

In this subsection, we construct $(m,1)_q$ convertible MSR array codes for the setting where the $m$ initial codewords are drawn from a common $(n_I,k_I;\alpha)_q$ MSR code $\mathcal{C}_I$. This setting is motivated by practical storage systems, which typically operate with a single uniformly deployed erasure code. When storage requirements change---due to data growth, node addition, or fault-tolerance adjustments---the system must convert multiple codewords of the same code into a new code. Concretely, the conversion procedure transforms $m$ codewords of $\mathcal{C}_I$ into a codeword of an $(n_F,k_F;\alpha)_q$ MSR code $\mathcal{C}_F$.

In the irregular setting, since the initial codes can be designed independently, the required periodic structure can be embedded at 
the design stage. However, when all initial codewords are drawn from the same code, the initial codes share a common structure that must be uniformly applied to all codewords. In this case, we introduce a row-matching technique that achieves the required alignment through a bijective row remapping in the conversion procedure itself, rather than at the design stage.

\subsubsection{Illustrative Examples}

The following example shows that the irregular construction in 
Section~\ref{sec: array} cannot produce two initial codewords from 
the same code, motivating the need for a different approach.

\begin{ex}\label{ex: challenges}
Consider $m=2$ with parameters $n_I=4$, $k_I=2$, $d_I=3$, $d_F=5$, $r_F=2$,
so that $s=\mathrm{lcm}(d_I-k_I+1, d_F-k_F+1)=2$ and $\alpha=2^6=64$. 
Let $\beta$ be a primitive element of $\mathbb{F}_q$. The evaluation 
points are $\lambda_{v,u}:=\beta^{2u+v}$ for $v\in\{0,1\}$ and 
$u\in[0,5]$. By Construction~\ref{cons: array_rI>r_F}, 
the two initial codes $\mathcal{C}_{I_1}$ and $\mathcal{C}_{I_2}$ 
have row coefficients
\begin{equation}\label{eq: c_1}
\lambda_{w_{(u)},\,u}=\beta^{2u+w_{(u)}},
\quad u\in[0,3],
\end{equation}
and
$
\lambda_{w_{(u+2)},\,u+2}=\beta^{2(u+2)+w_{(u+2)}},
\quad u\in[0,3],
$
respectively. Table~\ref{tab: ex_irregular} shows the row coefficients 
for the first few rows.

\begin{table}[H]
\centering
\caption{Row coefficients of $\mathcal{C}_{I_1}$ (left) and 
$\mathcal{C}_{I_2}$ (right) under Construction~\ref{cons: array_rI>r_F}.}
\label{tab: ex_irregular}
\vspace{-0.5cm}
\setlength{\tabcolsep}{1.5pt}
\renewcommand{\arraystretch}{0.01}
\begin{minipage}[t]{0.48\columnwidth}
\centering
\textbf{$\mathcal{C}_{I_1}$}\\[2pt]
\begin{tabular}{c|c|c|c|c}
\toprule
$w$ & node $0$ & node $1$ & node $2$ & node $3$ \\
\midrule
$0$ & {\color{red!70!black}$\beta^{0}$} & {\color{red!70!black}$\beta^{2}$} & {\color{red!70!black}$\beta^{4}$} & {\color{red!70!black}$\beta^{6}$}\\
$1$ & {\color{blue!70!black}$\beta^{1}$} & {\color{red!70!black}$\beta^{2}$} & {\color{red!70!black}$\beta^{4}$} & {\color{red!70!black}$\beta^{6}$} \\
$2$ & {\color{red!70!black}$\beta^{0}$} & {\color{blue!70!black}$\beta^{3}$} & {\color{red!70!black}$\beta^{4}$} & {\color{red!70!black}$\beta^{6}$} \\
$3$ & {\color{blue!70!black}$\beta^{1}$} & {\color{blue!70!black}$\beta^{3}$} & {\color{red!70!black}$\beta^{4}$} & {\color{red!70!black}$\beta^{6}$}\\
$4$ & {\color{red!70!black}$\beta^{0}$} & {\color{red!70!black}$\beta^{2}$} & {\color{blue!70!black}$\beta^{5}$} & {\color{red!70!black}$\beta^{6}$} \\
$\vdots$ & $\vdots$ & $\vdots$ & $\vdots$ & $\vdots$ \\
$8$ & {\color{red!70!black}$\beta^{0}$} & {\color{red!70!black}$\beta^{2}$} & {\color{red!70!black}$\beta^{4}$} & {\color{blue!70!black}$\beta^{7}$}\\
$\vdots$ & $\vdots$ & $\vdots$ & $\vdots$ & $\vdots$ \\
$16$ & {\color{red!70!black}$\beta^{0}$} & {\color{red!70!black}$\beta^{2}$} & {\color{red!70!black}$\beta^{4}$} & {\color{red!70!black}$\beta^{6}$}\\
$\vdots$ & $\vdots$ & $\vdots$ & $\vdots$ & $\vdots$ \\
$32$ & {\color{red!70!black}$\beta^{0}$} & {\color{red!70!black}$\beta^{2}$} & {\color{red!70!black}$\beta^{4}$} & {\color{red!70!black}$\beta^{6}$}\\
$\vdots$ & $\vdots$ & $\vdots$ & $\vdots$ & $\vdots$ \\
\bottomrule
\end{tabular}
\end{minipage}\hfill
\begin{minipage}[t]{0.08\columnwidth}
\centering
\vspace{2pt}
\end{minipage}\hfill
\begin{minipage}[t]{0.45\columnwidth}
\centering
\textbf{$\mathcal{C}_{I_2}$}\\[2pt]
\begin{tabular}{c|c|c|c}
\toprule
 node $0$ & node $1$ & node $2$ & node $3$ \\
\midrule
 {\color{red!70!black}$\beta^{4}$} & {\color{red!70!black}$\beta^{6}$} & {\color{red!70!black}$\beta^{8}$} & {\color{red!70!black}$\beta^{10}$} \\
 {\color{red!70!black}$\beta^{4}$} & {\color{red!70!black}$\beta^{6}$} & {\color{red!70!black}$\beta^{8}$} & {\color{red!70!black}$\beta^{10}$} \\
 {\color{red!70!black}$\beta^{4}$} & {\color{red!70!black}$\beta^{6}$} & {\color{red!70!black}$\beta^{8}$} & {\color{red!70!black}$\beta^{10}$} \\
 {\color{red!70!black}$\beta^{4}$} & {\color{red!70!black}$\beta^{6}$} & {\color{red!70!black}$\beta^{8}$} & {\color{red!70!black}$\beta^{10}$} \\
 {\color{blue!70!black}$\beta^{5}$} & {\color{red!70!black}$\beta^{6}$} & {\color{red!70!black}$\beta^{8}$} & {\color{red!70!black}$\beta^{10}$} \\
  $\vdots$ & $\vdots$ & $\vdots$ & $\vdots$ \\
 {\color{red!70!black}$\beta^{4}$} & {\color{blue!70!black}$\beta^{7}$} & {\color{red!70!black}$\beta^{8}$} & {\color{red!70!black}$\beta^{10}$} \\
 $\vdots$ & $\vdots$ & $\vdots$ & $\vdots$ \\
 {\color{red!70!black}$\beta^{4}$} & {\color{red!70!black}$\beta^{6}$} & {\color{blue!70!black}$\beta^{9}$} & {\color{red!70!black}$\beta^{10}$} \\
 $\vdots$ & $\vdots$ & $\vdots$ & $\vdots$ \\
 {\color{red!70!black}$\beta^{4}$} & {\color{red!70!black}$\beta^{6}$} & {\color{red!70!black}$\beta^{8}$} & {\color{blue!70!black}$\beta^{11}$} \\
 $\vdots$ & $\vdots$ & $\vdots$ & $\vdots$ \\
\bottomrule
\end{tabular}
\end{minipage}
\end{table}

We next explain why $\mathcal{C}_{I_1}$ and $\mathcal{C}_{I_2}$ 
cannot be chosen from the same MSR code under 
Construction~\ref{cons: array_rI>r_F}. If they were 
generated from the same MSR code, then each row of coefficients 
in $\mathcal{C}_{I_2}$ would have to be a nonzero scalar multiple 
of the corresponding row in $\mathcal{C}_{I_1}$, i.e., for each 
$w\in[0,63]$, there would exist $\delta_w\in\mathbb{F}_q^*$ such that
\[
\lambda_{w_{(u+2)},\,u+2}=\delta_w\lambda_{w_{(u)},\,u},
\quad \text{for all } u\in[0,3].
\]

However, this already fails for $w=1$. From node $u=0$ in row 
$w=1$ (where $w_{(0)}=1$ and $w_{(2)}=0$):
\[
\delta_1=\frac{\lambda_{0,\,2}}{\lambda_{1,\,0}}=
\frac{\beta^{4}}{\beta^{1}}=\beta^{3}.\]
From node $u=1$ in the same row (where $w_{(1)}=0$ and $w_{(3)}=0$):
\[
\delta_1=\frac{\lambda_{0,\,3}}{\lambda_{0,\,1}}=
\frac{\beta^{6}}{\beta^{2}}=\beta^{4},
\]
which contradicts $\delta_1=\beta^3$. Hence $\mathcal{C}_{I_1}$ 
and $\mathcal{C}_{I_2}$ cannot be the same MSR code under 
Construction~\ref{cons: array_rI>r_F}.
\end{ex}

To resolve this issue, we introduce the row-matching technique, 
which applies a bijective row remapping $\pi_i$ to each initial 
codeword before conversion. The key idea is to pair rows from the 
same MSR code across the initial codewords so that, after 
conversion, the resulting final code recovers the periodic 
structure required for MSR repair. The following example 
illustrates this approach.

\begin{ex}\label{ex: solution}
We use the same parameters as in Example~\ref{ex: challenges}. 
Both initial codewords are drawn from the same $(4,2;2^6)_q$ 
Hadamard MSR code $\mathcal{C}_I$. We take $\mathcal{C}_{I_1}$ 
to be the same as in Example~\ref{ex: challenges}, with row 
coefficients given by~\eqref{eq: c_1}. The second codeword is 
drawn from $\mathcal{C}'_{I_2}$, with row coefficients
\[
\lambda'_{w_{(u)},\,u}=\beta^4\cdot\lambda_{w_{(u)},\,u}
=\beta^{2u+w_{(u)}+4},\quad u\in[0,3].
\]

Since scaling all evaluation points by $\beta^4$ does not change 
the solution set of the parity-check equations, $\mathcal{C}_{I_1}$ 
and $\mathcal{C}'_{I_2}$ represent the same MSR code $\mathcal{C}_I$. 
Table~\ref{tab: ex_same_code} shows the row coefficients of the 
two initial codewords.

\begin{table}[htb]\footnotesize
\centering
\caption{Row coefficients of $\mathcal{C}_{I_1}$ (left) and 
$\mathcal{C}'_{I_2}$ (right).}
\label{tab: ex_same_code}
\vspace{-0.4cm}
\setlength{\tabcolsep}{1.5pt}
\renewcommand{\arraystretch}{0.01}
\begin{minipage}[t]{0.48\columnwidth}
\centering
\textbf{$\mathcal{C}_{I_1}$}\\[2pt]
\begin{tabular}{c|c|c|c|c}
\toprule
$w$ & node $0$ & node $1$ & node $2$ & node $3$ \\
\midrule
$0$ & {\color{red!70!black}$\beta^{0}$} & {\color{red!70!black}$\beta^{2}$} & {\color{red!70!black}$\beta^{4}$} & {\color{red!70!black}$\beta^{6}$}\\
$1$ & {\color{blue!70!black}$\beta^{1}$} & {\color{red!70!black}$\beta^{2}$} & {\color{red!70!black}$\beta^{4}$} & {\color{red!70!black}$\beta^{6}$} \\
$2$ & {\color{red!70!black}$\beta^{0}$} & {\color{blue!70!black}$\beta^{3}$} & {\color{red!70!black}$\beta^{4}$} & {\color{red!70!black}$\beta^{6}$} \\
$3$ & {\color{blue!70!black}$\beta^{1}$} & {\color{blue!70!black}$\beta^{3}$} & {\color{red!70!black}$\beta^{4}$} & {\color{red!70!black}$\beta^{6}$}\\
$4$ & {\color{red!70!black}$\beta^{0}$} & {\color{red!70!black}$\beta^{2}$} & {\color{blue!70!black}$\beta^{5}$} & {\color{red!70!black}$\beta^{6}$} \\
$\vdots$ & $\vdots$ & $\vdots$ & $\vdots$ & $\vdots$ \\
$8$ & {\color{red!70!black}$\beta^{0}$} & {\color{red!70!black}$\beta^{2}$} & {\color{red!70!black}$\beta^{4}$} & {\color{blue!70!black}$\beta^{7}$}\\
$\vdots$ & $\vdots$ & $\vdots$ & $\vdots$ & $\vdots$ \\
$16$ & {\color{red!70!black}$\beta^{0}$} & {\color{red!70!black}$\beta^{2}$} & {\color{red!70!black}$\beta^{4}$} & {\color{red!70!black}$\beta^{6}$}\\
$\vdots$ & $\vdots$ & $\vdots$ & $\vdots$ & $\vdots$ \\
$32$ & {\color{red!70!black}$\beta^{0}$} & {\color{red!70!black}$\beta^{2}$} & {\color{red!70!black}$\beta^{4}$} & {\color{red!70!black}$\beta^{6}$}\\
$\vdots$ & $\vdots$ & $\vdots$ & $\vdots$ & $\vdots$ \\
\bottomrule
\end{tabular}
\end{minipage}\hfill
\begin{minipage}[t]{0.45\columnwidth}
\centering
\textbf{$\mathcal{C}'_{I_2}$}\\[0.6pt]
\begin{tabular}{c|c|c|c|c}
\toprule
$w$ & node $0$ & node $1$ & node $2$ & node $3$ \\
\midrule
$0$ & {\color{red!70!black}$\beta^{4}$} & {\color{red!70!black}$\beta^{6}$} & {\color{red!70!black}$\beta^{8}$} & {\color{red!70!black}$\beta^{10}$} \\
$1$ & {\color{blue!70!black}$\beta^{5}$} & {\color{red!70!black}$\beta^{6}$} & {\color{red!70!black}$\beta^{8}$} & {\color{red!70!black}$\beta^{10}$} \\
$2$ & {\color{red!70!black}$\beta^{4}$} & {\color{blue!70!black}$\beta^{7}$} & {\color{red!70!black}$\beta^{8}$} & {\color{red!70!black}$\beta^{10}$} \\
$3$ & {\color{blue!70!black}$\beta^{5}$} & {\color{blue!70!black}$\beta^{7}$} & {\color{red!70!black}$\beta^{8}$} & {\color{red!70!black}$\beta^{10}$} \\
$4$ & {\color{red!70!black}$\beta^{4}$} & {\color{red!70!black}$\beta^{6}$} & {\color{blue!70!black}$\beta^{9}$} & {\color{red!70!black}$\beta^{10}$} \\
$\vdots$ & $\vdots$ & $\vdots$ & $\vdots$ & $\vdots$ \\
$8$ & {\color{red!70!black}$\beta^{4}$} & {\color{red!70!black}$\beta^{6}$} & {\color{red!70!black}$\beta^{8}$} & {\color{blue!70!black}$\beta^{11}$} \\
$\vdots$ & $\vdots$ & $\vdots$ & $\vdots$ & $\vdots$ \\
$16$ & {\color{red!70!black}$\beta^{4}$} & {\color{red!70!black}$\beta^{6}$} & {\color{red!70!black}$\beta^{8}$} & {\color{red!70!black}$\beta^{10}$} \\
$\vdots$ & $\vdots$ & $\vdots$ & $\vdots$ & $\vdots$ \\
$32$ & {\color{red!70!black}$\beta^{4}$} & {\color{red!70!black}$\beta^{6}$} & {\color{red!70!black}$\beta^{8}$} & {\color{red!70!black}$\beta^{10}$} \\
$\vdots$ & $\vdots$ & $\vdots$ & $\vdots$ & $\vdots$ \\
\bottomrule
\end{tabular}
\end{minipage}
\end{table}

To preserve the coordinate-wise periodic structure in the final 
code, the rows of the two initial codewords are paired as follows. 
The rows of $\mathcal{C}_{I_1}$ are kept unchanged, while the 
rows of $\mathcal{C}'_{I_2}$ are cyclically shifted by $k_I=2$ 
positions in their $2$-ary expansion before merging. Specifically, 
for a given row index $w$ of the final code, the unchanged symbols 
of $\mathcal{C}_{I_1}$ are taken from row $w$, while those of 
$\mathcal{C}'_{I_2}$ are taken from the row whose $2$-ary 
coordinates are shifted by $2$ positions, i.e., from row 
$w'=(w_{(2)},w_{(3)},w_{(4)},w_{(5)},w_{(0)},w_{(1)})$.

As a result, the row coefficient at node 
$u\in[0,3]$ of $\mathcal{C}'_{I_2}$, which originally depends 
on $w_{(u)}$, now depends on $w_{(u+2)}$. More precisely, the 
coefficient at node $u$ under the shifted row $w'$ satisfies
\[\small
\lambda'_{w'_{(u)},\,u}=\beta^{2u+w'_{(u)}+4}
=\beta^{2u+w_{(u+2)}+4},\quad u\in[0,3],
\]
so that the dependence on $w$ shifts from coordinate $u$ to 
coordinate $u+2$ for all $u\in[0,3]$. 
Table~\ref{tab: ex_matching} shows the row coefficients of the 
two initial codewords after applying the cyclic shift.

\begin{table}[H]
\centering
\caption{Row coefficients of $\mathcal{C}_{I_1}$ (left) and 
 $\mathcal{C}'_{I_2}$ (right) after  cyclic shift.}
 \vspace{-0.5cm}
\label{tab: ex_matching}
\setlength{\tabcolsep}{1.1pt}
\renewcommand{\arraystretch}{0.01}
\begin{minipage}[t]{0.4\columnwidth}
\centering
\textbf{$\mathcal{C}_{I_1}$, row $w$}\\[2pt]
\begin{tabular}{c|c|c|c|c}
\toprule
$w$ & node $0$ & node $1$ & node $2$ & node $3$ \\
\midrule
$0$ & {\color{red!70!black}$\beta^{0}$} & {\color{red!70!black}$\beta^{2}$} & {\color{red!70!black}$\beta^{4}$} & {\color{red!70!black}$\beta^{6}$}\\
$1$ & {\color{blue!70!black}$\beta^{1}$} & {\color{red!70!black}$\beta^{2}$} & {\color{red!70!black}$\beta^{4}$} & {\color{red!70!black}$\beta^{6}$} \\
$2$ & {\color{red!70!black}$\beta^{0}$} & {\color{blue!70!black}$\beta^{3}$} & {\color{red!70!black}$\beta^{4}$} & {\color{red!70!black}$\beta^{6}$} \\
$3$ & {\color{blue!70!black}$\beta^{1}$} & {\color{blue!70!black}$\beta^{3}$} & {\color{red!70!black}$\beta^{4}$} & {\color{red!70!black}$\beta^{6}$}\\
$4$ & {\color{red!70!black}$\beta^{0}$} & {\color{red!70!black}$\beta^{2}$} & {\color{blue!70!black}$\beta^{5}$} & {\color{red!70!black}$\beta^{6}$} \\
$\vdots$ & $\vdots$ & $\vdots$ & $\vdots$ & $\vdots$ \\
$8$ & {\color{red!70!black}$\beta^{0}$} & {\color{red!70!black}$\beta^{2}$} & {\color{red!70!black}$\beta^{4}$} & {\color{blue!70!black}$\beta^{7}$}\\
$\vdots$ & $\vdots$ & $\vdots$ & $\vdots$ & $\vdots$ \\
$16$ & {\color{red!70!black}$\beta^{0}$} & {\color{red!70!black}$\beta^{2}$} & {\color{red!70!black}$\beta^{4}$} & {\color{red!70!black}$\beta^{6}$}\\
$\vdots$ & $\vdots$ & $\vdots$ & $\vdots$ & $\vdots$ \\
$32$ & {\color{red!70!black}$\beta^{0}$} & {\color{red!70!black}$\beta^{2}$} & {\color{red!70!black}$\beta^{4}$} & {\color{red!70!black}$\beta^{6}$}\\
$\vdots$ & $\vdots$ & $\vdots$ & $\vdots$ & $\vdots$ \\
\bottomrule
\end{tabular}
\end{minipage}\hfill
\begin{minipage}[t]{0.5\columnwidth}
\centering
\textbf{$\mathcal{C}'_{I_2}$, row $w'$}\\[0.2pt]
\begin{tabular}{c|c|c|c|c}
\toprule
$w'$ & node $0$ & node $1$ & node $2$ & node $3$ \\
\midrule
$0$ & {\color{red!70!black}$\beta^{4}$} & {\color{red!70!black}$\beta^{6}$} & {\color{red!70!black}$\beta^{8}$} & {\color{red!70!black}$\beta^{10}$} \\
$16$ & {\color{red!70!black}$\beta^{4}$} & {\color{red!70!black}$\beta^{6}$} & {\color{red!70!black}$\beta^{8}$} & {\color{red!70!black}$\beta^{10}$} \\
$32$ & {\color{red!70!black}$\beta^{4}$} & {\color{red!70!black}$\beta^{6}$} & {\color{red!70!black}$\beta^{8}$} & {\color{red!70!black}$\beta^{10}$} \\
$48$ & {\color{red!70!black}$\beta^{4}$} & {\color{red!70!black}$\beta^{6}$} & {\color{red!70!black}$\beta^{8}$} & {\color{red!70!black}$\beta^{10}$} \\
$1$ & {\color{blue!70!black}$\beta^{5}$} & {\color{red!70!black}$\beta^{6}$} & {\color{red!70!black}$\beta^{8}$} & {\color{red!70!black}$\beta^{10}$} \\
$\vdots$ & $\vdots$ & $\vdots$ & $\vdots$ & $\vdots$ \\
$2$ & {\color{red!70!black}$\beta^{4}$} & {\color{blue!70!black}$\beta^{7}$} & {\color{red!70!black}$\beta^{8}$} & {\color{red!70!black}$\beta^{10}$} \\
$\vdots$ & $\vdots$ & $\vdots$ & $\vdots$ & $\vdots$ \\
$4$ & {\color{red!70!black}$\beta^{4}$} & {\color{red!70!black}$\beta^{6}$} & {\color{blue!70!black}$\beta^{9}$} & {\color{red!70!black}$\beta^{10}$} \\
$\vdots$ & $\vdots$ & $\vdots$ & $\vdots$ & $\vdots$ \\
$8$ & {\color{red!70!black}$\beta^{4}$} & {\color{red!70!black}$\beta^{6}$} & {\color{red!70!black}$\beta^{8}$} & {\color{blue!70!black}$\beta^{11}$} \\
$\vdots$ & $\vdots$ & $\vdots$ & $\vdots$ & $\vdots$ \\
\bottomrule
\end{tabular}
\end{minipage}
\end{table}

Observe that in the row-permuted $\mathcal{C}'_{I_2}$, the 
coefficient at  unchanged node $0$ now depends on $w_{(2)}$ rather than 
$w_{(0)}$, and the coefficient at unchanged node $1$ depends on $w_{(3)}$ 
rather than $w_{(1)}$. After applying the conversion framework, 
the final code has row coefficients as shown in 
Table~\ref{tab: ex_solution}.

\begin{table}[htb]
\centering
\caption{Row coefficients of the final code after row-matching 
conversion.}
\vspace{-0.25cm}
\setlength{\tabcolsep}{3pt}
\renewcommand{\arraystretch}{0.001} 
\label{tab: ex_solution}
\begin{tabular}{c|c|c|c|c|c|c}
\toprule
$w$ & node $0$ & node $1$ & node $2$ & node $3$ & node $4$ & node $5$\\
\midrule
$0$ & {\color{red!70!black}$\beta^{0}$} & {\color{red!70!black}$\beta^{2}$} & 
{\color{red!70!black}$\beta^{4}$} & {\color{red!70!black}$\beta^{6}$} & 
{\color{red!70!black}$\beta^{8}$} & {\color{red!70!black}$\beta^{10}$} \\
$1$ & {\color{blue!70!black}$\beta^{1}$} & {\color{red!70!black}$\beta^{2}$} & 
{\color{red!70!black}$\beta^{4}$} & {\color{red!70!black}$\beta^{6}$} & 
{\color{red!70!black}$\beta^{8}$} & {\color{red!70!black}$\beta^{10}$} \\
$2$ & {\color{red!70!black}$\beta^{0}$} & {\color{blue!70!black}$\beta^{3}$} & 
{\color{red!70!black}$\beta^{4}$} & {\color{red!70!black}$\beta^{6}$} & 
{\color{red!70!black}$\beta^{8}$} & {\color{red!70!black}$\beta^{10}$} \\
$3$ & {\color{blue!70!black}$\beta^{1}$} & {\color{blue!70!black}$\beta^{3}$} & 
{\color{red!70!black}$\beta^{4}$} & {\color{red!70!black}$\beta^{6}$} & 
{\color{red!70!black}$\beta^{8}$} & {\color{red!70!black}$\beta^{10}$} \\
$4$ & {\color{red!70!black}$\beta^{0}$} & {\color{red!70!black}$\beta^{2}$} & 
{\color{blue!70!black}$\beta^{5}$} & {\color{red!70!black}$\beta^{6}$} & 
{\color{red!70!black}$\beta^{8}$} & {\color{red!70!black}$\beta^{10}$} \\
$\vdots$ & $\vdots$ & $\vdots$ & $\vdots$ & $\vdots$ & $\vdots$ & $\vdots$\\
$8$ & {\color{red!70!black}$\beta^{0}$} & {\color{red!70!black}$\beta^{2}$} & 
{\color{red!70!black}$\beta^{4}$} & {\color{blue!70!black}$\beta^{7}$} & 
{\color{red!70!black}$\beta^{8}$} & {\color{red!70!black}$\beta^{10}$} \\
$\vdots$ & $\vdots$ & $\vdots$ & $\vdots$ & $\vdots$ & $\vdots$ & $\vdots$\\
$16$ & {\color{red!70!black}$\beta^{0}$} & {\color{red!70!black}$\beta^{2}$} & 
{\color{red!70!black}$\beta^{4}$} & {\color{red!70!black}$\beta^{6}$} & 
{\color{blue!70!black}$\beta^{9}$} & {\color{red!70!black}$\beta^{10}$} \\
$\vdots$ & $\vdots$ & $\vdots$ & $\vdots$ & $\vdots$ & $\vdots$ & $\vdots$\\
$32$ & {\color{red!70!black}$\beta^{0}$} & {\color{red!70!black}$\beta^{2}$} & 
{\color{red!70!black}$\beta^{4}$} & {\color{red!70!black}$\beta^{6}$} & 
{\color{red!70!black}$\beta^{8}$} & {\color{blue!70!black}$\beta^{11}$} \\
$\vdots$ & $\vdots$ & $\vdots$ & $\vdots$ & $\vdots$ & $\vdots$ & $\vdots$\\
\bottomrule
\end{tabular}
\end{table}

As shown in Table~\ref{tab: ex_solution}, the row coefficient at 
each node $u\in[0,5]$ of the final code is
$
\alpha_{w,u}:=\lambda_{w_{(u)},u}=\beta^{2u+w_{(u)}},
$
which depends on $w$ only through its $u$-th coordinate in the 
$2$-ary expansion. This is precisely the same subsymbol-level 
alignment as achieved by the pre-alignment in 
Construction~\ref{cons: array_rI>r_F}: the cyclic shift of $\mathcal{C}'_{I_2}$ 
by $k_I=2$ positions produces exactly the coordinate  
$w_{(u+2)}$ at nodes $u\in[0,3]$, matching the structure 
of~\eqref{eq: pc_initial}. Hence, the final code is an 
$(n_F=6,k_F=4;2^6)_q$ MSR code by the same argument as 
Theorem~\ref{thm: msr}.
\end{ex}

The examples above illustrate the key ideas behind the row-matching 
technique. We now formalize this approach.

\subsubsection{The Row-Matching Technique}\label{sec: row-matching map}
Let $\mathcal{C}_{I}$ be the initial code, which is an
$(n_I,k_I;s^L)_q$ Hadamard-based MSR code whose row coefficients are given by
$
\alpha_{w,u}:=\lambda_{w_{(u)},u},\quad u\in[0,n_I-1],\,w\in[0,s^L-1].
$

\textbf{Step 1: Scale the evaluation points.}
For each $i\in[m]$, define a code $\mathcal{C}_{I_i}$ with row 
coefficients
$
\delta_i\alpha_{w,u}=\delta_i\lambda_{w_{(u)},u},\,
u\in[0,n_I-1],\,w\in[0,s^L-1],
$
where $\delta_i\in\mathbb{F}_q^*$ is a nonzero scalar such that 
the scaled sets
\begin{equation}\label{eq: scaled_Set}
\Lambda^{(i)}_{k_I}:=\{\delta_i\lambda_{v,u}:v\in[0,s-1],\,
u\in[0,k_I-1]\}
\end{equation}
are pairwise disjoint. Then, the parity-check equations of 
$\mathcal{C}_{I_i}$ are given by
\begin{equation*}\small
\begin{split}
\sum_{u=0}^{n_I-1}(\delta_i\lambda_{w_{(u)},u})^t c^{(i)}_{w,u}=\delta_i^t\sum_{u=0}^{n_I-1}\lambda_{w_{(u)},u}^t c^{(i)}_{w,u}=0,\\
\quad t\in[0,r_I-1],\,w\in[0,s^L-1],
\end{split}
\end{equation*}
which differ from those of $\mathcal{C}_I$ only by the common 
scaling factor $\delta_i^t$. Since multiplying all evaluation 
points by a common nonzero scalar does not change the solution 
set of the parity-check equations, all $\mathcal{C}_{I_i}$ 
represent the same MSR code $\mathcal{C}_I$. Furthermore, the 
disjointness condition on the first $k_I$ evaluation points 
ensures that the unchanged symbols from different initial codes 
do not collide, which is required for the MDS property of the 
final code.

\textbf{Step 2: Define the row-matching maps.}
For each $i\in[m]$, let $\pi_i:[0,\alpha-1]\to[0,\alpha-1]$ be 
a bijection. We set $\pi_1(w)=w$, leaving the first codeword 
unchanged, and for $i\geq 2$ define $\pi_i$ as a cyclic left-shift 
of the $s$-ary expansion of $w$ by $(i-1)k_I$ positions:
\begin{equation}\label{eq: row_matching mapping}
\begin{split}
\pi_i(w_{(0)},\cdots,&w_{(L-1)}):=\\
(w_{((i-1)k_I)},\cdots,&w_{(L-1)},w_{(0)},\cdots,w_{((i-1)k_I-1)}).
\end{split}
\end{equation}
Since each $\pi_i$ is induced by a coordinate permutation of the 
$s$-ary expansion of $w$, it is a bijection on $[0,\alpha-1]$. 
By~\eqref{eq: row_matching mapping}, the $u$-th coordinate of 
$\pi_i(w)$ satisfies $\pi_i(w)_{(u)}=w_{(u+(i-1)k_I)}$, so the 
parity-check equations of the $i$-th codeword at row $\pi_i(w)$ 
take the form
\vspace{-0.2cm}
\[\small
\sum_{u=0}^{n_I-1}\lambda_{w_{(u+(i-1)k_I)},\,u}^t\,
c^{(i)}_{\pi_i(w),u}=0,
\]
which is precisely~\eqref{eq: pc_initial} in 
Construction~\ref{cons: array_rI>r_F}. Hence the row-matching converts the same-code 
conversion into an instance of Construction~\ref{cons: array_rI>r_F}, 
and the MSR repair properties of both the initial codes and the 
final code follow directly from Theorem~\ref{thm: msr}.

\textbf{Step 3: Apply the conversion framework.}
For each $w\in[0,\alpha-1]$, the $m$ rows
\[\small
\big(c^{(i)}_{\pi_i(w),0},\cdots,c^{(i)}_{\pi_i(w),n_I-1}\big),
\quad i\in[m],
\]
are grouped together to form one scalar conversion instance, and 
the remaining conversion steps follow exactly those in 
Section~\ref{sec: framework}.
\subsubsection{Same-Code Convertible MSR Codes}

The row-matching technique applies to any Hadamard-based MSR code 
$\mathcal{C}_I$ with sufficiently large $L$ and field size $q$, 
yielding convertible MSR codes in which all initial codewords are 
drawn from the same code. We now instantiate this procedure using 
Construction~\ref{cons: array_rI>r_F} and 
Construction~1A. 
As we will show, the same-code constraint does not incur any loss: 
the resulting codes achieve optimal access cost and 
optimal conversion bandwidth in the corresponding parameter regimes.

Since all initial codewords are drawn from $\mathcal{C}_I$, the 
parameters satisfy $n_{I_i}=n_I$, $k_{I_i}=k_I$, and $r_{I_i}=r_I$ 
for all $i\in[m]$, and the sub-packetization parameter simplifies to 
$s=\mathrm{lcm}(d_I-k_I+1,\,d_F-k_F+1)$. We next specialize the above three-step procedure to the same-code setting.

\textbf{Step 1} (Scaling). We take $\mathcal{C}_I=\mathcal{C}_{I_1}$, thus
row coefficients are given by
\begin{equation}\label{eq: c_1_coefficient}
\alpha_{w,u}:=\lambda_{w_{(u)},u}=\beta^{su+w_{(u)}},\quad u\in[0,n_I-1].
\end{equation}

For each $i\in[m]$, set $\delta_i:=\beta^{(i-1)sk_I}$. The row coefficients of $\mathcal{C}_{I_i}$ 
are then
\[
\delta_i\lambda_{w_{(u)},u}=\beta^{s((i-1)k_I+u)+w_{(u)}},
\quad u\in[0,n_I-1],
\]
and all $\mathcal{C}_{I_i}$ represent the 
same MSR code $\mathcal{C}_I$.
By \eqref{eq: scaled_Set} and \eqref{eq: c_1_coefficient}, we can deduce that
the scaled sets
\[\small
\Lambda_{k_I}^{(i)}=\bigl\{\delta_i\beta^{su+v}:v\in[0,s-1],\,
u\in[0,k_I-1]\bigr\}
\]
are pairwise disjoint. 

\textbf{Step 2} (Row-matching). For each $i\in[m]$, apply the 
bijection $\pi_i$ defined in~\eqref{eq: row_matching mapping} to 
the rows of $\mathcal{C}_{I_i}$.

\textbf{Step 3} (Conversion). For each $w\in[0,\alpha-1]$, row 
$\pi_i(w)$ of $\mathcal{C}_{I_i}$ serves as the $i$-th input to 
the conversion procedure in Section~\ref{sec: framework}, with 
$\hat{H}_w$ given by~\eqref{eq: cons1_h} for 
Construction~\ref{cons: array_rI>r_F} and by~\eqref{eq: cons2_h} 
for Construction~1A. The parity-check 
equations of the resulting final code $\mathcal{C}_F$ are
\begin{equation}\label{eq: step3_final}\small
\begin{split}
\sum_{i=1}^m\sum_{u=0}^{k_I-1}\!&\left(\beta^{s((i-1)k_I+u)
+\pi_i(w)_{(u)}}\right)^{\!t}\!c_{w,(i-1)k_I+u}
+\\&\sum_{u=0}^{r_F-1}\!\left(\beta^{s(mk_I+u)+w_{(mk_I+u)}}
\right)^{\!t}\!c_{w,mk_I+u}=0,
\end{split}
\end{equation}
for $t\in[0,r_F-1]$ and $w\in[0,\alpha-1]$. By 
\eqref{eq: row_matching mapping}, the row-matching bijection 
satisfies
\[
\pi_i(w)_{(u)}=w_{((i-1)k_I+u)},\quad i\in[m],\,u\in[0,n_I-1].
\]
Substituting into~\eqref{eq: step3_final} yields
\[\small
\sum_{u=0}^{n_F-1}\left(\beta^{su+w_{(u)}}\right)^t c_{w,u}=0,
\quad t\in[0,r_F-1],\;w\in[0,\alpha-1],
\]
which is precisely the Hadamard MSR structure 
in~\eqref{eq: had_row_points_simple}. Hence $\mathcal{C}_F$ is 
an MSR code.

Moreover, since all initial codes share the same 
parameters, the conversion bandwidth lower bound in 
Theorem~\ref{bound on bandwidth} applies. 
Construction~\ref{cons: array_rI>r_F} achieves this bound for 
$r_F\leq\min\{k_I,r_I\}$, and Construction~1A
achieves it for $r_F>k_I$. We summarize as follows.

\begin{thm}\label{thm: same_code_access_bw}
Applying the row-matching technique to Construction~\ref{cons: array_rI>r_F} in the regime $r_F\leq \min\{k_I,r_I\}$ and to Construction~1A in the regime $r_F>k_I$ yields $(m,1)_q$ convertible MSR codes with optimal access cost and optimal conversion bandwidth.
\end{thm}
\begin{IEEEproof}
The MSR property of $\mathcal{C}_F$ follows from the derivation 
in Step~3. For Construction~\ref{cons: array_rI>r_F} in the regime 
$r_F\leq\min\{k_I,r_I\}$, the access cost equals $(m+1)r_F$, 
meeting the lower bound in Theorem~\ref{bound on access}, and the 
conversion bandwidth equals $(m+1)r_F\alpha$, 
meeting the lower bound in Theorem~\ref{bound on bandwidth}. The 
same holds for Construction~1A in the regime 
$r_F>k_I$, where the access cost equals $r_F+mk_I$ and the 
conversion bandwidth equals $r_F\alpha+mk_I\alpha$.
\end{IEEEproof}

\begin{remark}
While $\sigma_1$ and $\sigma_2$ achieve access optimality in the regimes 
$r_F\leq\min\{k_I,r_I\}$ and $r_F>\min\{k_I,r_I\}$, respectively, 
they do not necessarily achieve bandwidth optimality in all regimes.

In particular, in the regime $r_I<r_F<k_I$, the 
access-optimal construction uses $\sigma_2$, which has conversion bandwidth 
$r_F\alpha+mk_I\alpha$. This exceeds the lower bound
\begin{equation}\label{eq: acc+ban}
\gamma\geq r_F\alpha+m\alpha(r_I+k_I(1-\frac{r_I}{r_F})).
\end{equation}
This gap arises because $\sigma_2$ computes the $r_F$ written symbols 
by downloading all $k_I$ symbols from each initial codeword.

\end{remark}

To achieve the bound~\eqref{eq: acc+ban}, we instead use $\sigma_1$ augmented with the multi-node 
repair scheme of Section~\ref{sec: hadamard}(ii), as detailed in 
Stages 1 and 2 of Construction~\ref{cons: array_bw} below. This 
reduces the conversion bandwidth at the cost of reading additional 
symbols beyond the $mk_I$ nodes. In fact, achieving 
the bandwidth lower bound~\eqref{eq: acc+ban} in this regime 
inherently requires contacting more than $mk_I$ nodes, which 
exceeds the access lower bound. Hence access-optimality and 
bandwidth-optimality are fundamentally incompatible in this regime, 
and we focus on bandwidth-optimal constructions. In particular, 
we set $d_F=k_F+r_I$.

Since this construction uses the multi-node repair scheme of 
Section~\ref{sec: hadamard}(ii), a space-sharing extension of the 
subpacketization is required. We index each row by a pair $(w,z)$ 
with $w\in[0,s^L-1]$ and $z\in[0,r_F-1]$, so that $\alpha=r_Fs^L$. 
The index $w$ carries the periodic structure required for MSR repair 
and row-matching, while $z$ indexes the space-sharing instances. 
The row-matching maps $\pi_i$ act only on $w$, leaving $z$ unchanged.

\begin{cons}\label{cons: array_bw}
Consider the regime $r_I<r_F<k_I$. Let $\mathcal{C}_I$ be an $(n_I,k_I;\alpha)_q$ initial code. Define $k_I<d_I< n_I$
 and $d_F=k_F+r_I$ as the number of helper nodes of $\mathcal{C}_I$ and $\mathcal{C}_F$, respectively.
\begin{itemize}
    \item {Parameters.} Define
    \[\small s:=\mathrm{lcm}(d_I-k_I+1,\,d_F-k_F+1)=\mathrm{lcm}(d_I-k_I+1,\,r_I+1)\]
    and $L:=mk_I+r_F=n_F,\, \alpha:=r_Fs^{L}.$ Assume that $q>Ls$.

    \item {Initial code.} Let $\beta$ be a primitive element of 
    $\mathbb{F}_q$. Define
    \[\lambda_{v,u}:=\beta^{su+v},\quad v\in[0,s-1],\,u\in[0,n_I-1].\]

    Node $u$ of the $i$-th codeword of $\mathcal{C}_I$ stores
    \begin{equation*}\small
    \begin{split}
    \mathbf{c}_{i,u}=\bigl(c^{(i)}_{w,z,u}:\,w\in[0,s^L-1],\,
    z\in[0,&r_F-1]\bigr)^\top,\quad\\ &u\in[0,n_I-1].
    \end{split}
    \end{equation*}
    with the same parity-check equations \vspace{-0.2cm}
    \begin{equation}\label{eq: pc_initial_bw}
    \begin{split}
    \sum_{u=0}^{n_I-1}\left(\lambda_{w_{(u)},\,u}\right)^t
    &c^{(i)}_{w,z,u}=0,\quad \\t\in[0,r_I-1],\,&w\in[0,s^L-1],\,z\in[0,r_F-1].
    \end{split}
    \end{equation}
    
    \item {Conversion.} 
    
    Apply the row-matching technique in 
    Section~\ref{sec: row-matching map} to $\mathcal{C}_I$, 
    where $\pi_i$ acts on the index $w$ only. For each fixed 
    $(w,z)$, the unchanged symbols from row $(\pi_i(w),z)$ of 
    the $i$-th codeword serve as input to the conversion. 
    
 For any $z\in[0,r_F-1]$, the parity-check matrix $\hat{H}_{w,z}$ corresponding written symbols is
        given by
        \begin{equation}\label{eq: cons3_h}
        \begin{split}
        \hat{H}_{w,z}=\mathrm{Vand}_{r_F}\!\left(\lambda_{w_{(u)},\,u}:
        u\in[mk_I,\,mk_I+r_F-1]\right), \\ w\in[0,s^L-1].
        \end{split}
        \end{equation}

    The 
    written symbols
    \begin{equation*}\small
    \begin{split}
    \hat{\mathbf{c}}_u=\bigl(\hat{c}_{w,z,u}:w\in[0,s^L-1],\,
    z\in[0,\,r_F&-1]\bigr)^\top,\quad\\& u\in[0,r_F-1],
    \end{split}
    \end{equation*}
    are computed in two stages:
  \begin{enumerate}
      \item[Stage 1.] The first $r_I$ written symbols 
$\hat{\mathbf{c}}_0,\cdots,\hat{\mathbf{c}}_{r_I-1}$ are generated 
via $\sigma_1$. After this stage, the $mk_I$ unchanged symbols 
together with $\hat{\mathbf{c}}_0,\cdots,\hat{\mathbf{c}}_{r_I-1}$ 
form $k_F+r_I=d_F$ known nodes of the final code $\mathcal{C}_F$, 
which serve as helper nodes for Stage 2.

\item[Stage 2.] The remaining $r_F-r_I$ written symbols 
$\hat{\mathbf{c}}_{r_I},\cdots,\hat{\mathbf{c}}_{r_F-1}$ are 
treated as $r_F-r_I$ simultaneously failed nodes of the final 
code. As we will show in Theorem~\ref{thm: array_bw}, the final 
code $\mathcal{C}_F$ has the Hadamard MSR structure 
in~\eqref{eq: had_row_points_simple} with parameter $s_F=r_I+1$. 
The multi-node repair scheme of Section~\ref{sec: hadamard}(ii) 
is therefore applicable with $h=r_F-r_I$ failed nodes and 
$d_F=k_F+r_I$ helper nodes, yielding the $r_F-r_I$ unknown 
written symbols.
   \end{enumerate}
\end{itemize}
\end{cons}

\begin{thm}\label{thm: array_bw}
Construction~\ref{cons: array_bw} yields an $(m,1)_q$ convertible 
MSR code with optimal conversion bandwidth. The initial code $\mathcal{C}_I$ is an  $(n_I,k_I;\alpha)_q$ MSR code, and the final code $\mathcal{C}_F$ is an $(n_F,k_F=mk_I;\alpha)_q$ MSR code.
\end{thm}

\begin{IEEEproof}
We verify that Construction~\ref{cons: array_bw} yields a 
convertible MSR code with optimal conversion bandwidth by 
establishing the MSR property of both the initial code and the 
final code, and then proving bandwidth optimality.

\textbf{Initial code is MSR.} 
Fix any $z\in[0,r_F-1]$ and treat 
$\{c^{(i)}_{w,z,u}:w\in[0,s^L-1]\}$ as the subsymbols stored in 
node $u$. The parity-check equations~\eqref{eq: pc_initial_bw} 
have exactly the Hadamard structure in~\eqref{eq: had_pc_row} 
and~\eqref{eq: had_row_points_simple}, with subpacketization $s^L$ 
and parameter $s_I:=d_I-k_I+1$ dividing $s$ by definition. For 
each fixed $z$, the repair argument is analogous to that in 
Theorem~\ref{thm: msr}: for each failed node $j^*\in[0,n_I-1]$ 
and each interval $a\in[0,s/s_I-1]$, define the sub-group
\[\small
\mathcal{G}_a(w,j^*):=\bigl\{R_w(j^*,\,as_I+b):b\in[0,s_I-1]\bigr\},
\]
and apply the repair procedure of Section~\ref{sec: hadamard}(i) 
within each sub-group, 
giving per-helper download $s^L/s_I$ subsymbols. Running over all $r_F$ instances indexed by 
$z\in[0,r_F-1]$, the total repair bandwidth of each helper is
\[
b=r_F\cdot\frac{s^L}{s_I}=\frac{\alpha}{s_I}=
\frac{\alpha}{d_I-k_I+1},
\]
meeting the cut-set bound in~\eqref{eq: beta}. Hence 
$\mathcal{C}_I$ is an $(n_I,k_I;\alpha)_q$ MSR code.

\textbf{Final code is MSR.}
 For $w\in[0,s^L-1]$ and $z\in[0,r_F-1]$, we index the symbols 
of the final codeword as
\begin{equation*}
\begin{split}
&c_{w,z,\,j+\sum_{l=1}^{i-1}k_I}:=c^{(i)}_{\pi_i(w),z,j},
\quad i\in[m],\,j\in[0,k_I-1],\\
&c_{w,z,\,j+mk_I}:=\hat{c}_{w,z,j},\quad j\in[0,r_F-1],
\end{split}
\end{equation*}
where the first $mk_I$ positions correspond to the unchanged 
symbols from the $m$ initial codewords after row-matching, and 
the last $r_F$ positions correspond to the written symbols.

The written symbols are computed in two stages.
\begin{enumerate}
\item[1)] In the first stage, 
$\sigma_1$ is applied to generate the first $r_I$ written symbols 
$\hat{\mathbf{c}}_0,\cdots,\hat{\mathbf{c}}_{r_I-1}$ from the 
reading $r_I$ symbols of the initial codewords, i.e., 
$\{\mathbf{c}_{i,j}: j\in[k_I,k_I+r_I-1] \},\, i\in[m]$
with parity-check matrix 
\begin{equation*}
\begin{split}
        \hat{H}'_{w,z}=\mathrm{Vand}_{r_I}\!\left(\lambda_{w_{(u)},\,u}: u\in[mk_I,\,mk_I+r_I-1]\right),\\
        w\in[0,s^L-1],\quad z\in[0,r_F-1],
        \end{split}
        \end{equation*}
        which are $r_I\times r_I$ invertible submatrices of \eqref {eq: cons3_h}.
        
\item[2)] In the second stage, the remaining $r_F-r_I$ written 
symbols $\hat{\mathbf{c}}_{r_I},\cdots,\hat{\mathbf{c}}_{r_F-1}$ 
are treated as $h=r_F-r_I$ failed nodes, and are recovered from 
the $d_F=k_F+r_I$ helper nodes consisting of the unchanged symbols 
and $\hat{\mathbf{c}}_0,\cdots,\hat{\mathbf{c}}_{r_I-1}$ via the 
multi-node repair procedure of Section~\ref{sec: hadamard}(ii). This is valid because the parity-check equations of $\mathcal{C}_F$ 
satisfy the coordinate-wise periodicity required for MSR repair. 
Specifically, by the row-matching bijection, 
the $u$-th coordinate of $\pi_i(w)$ satisfies
\[
\pi_i(w)_{(u)}=w_{(u+(i-1)k_I)},\quad u\in[0,n_I-1],\,i\in[m],
\]
so the coefficient at node $u+(i-1)k_I$ of the final code, 
contributed by the $i$-th initial codeword, is
\begin{equation*}
\begin{split}
\delta_i\lambda_{\pi_i(w)_{(u)},\,u}&=\delta_i\lambda_{w_{(u+(i-1)k_I)},\,u}\\
&=\beta^{s(u+(i-1)k_I)+w_{(u+(i-1)k_I)}}.
\end{split}
\end{equation*}
Similarly, by~\eqref{eq: cons3_h}, the written symbols at node $mk_I+j$ for $j\in[0,r_F-1]$ 
have coefficient $\beta^{s(mk_I+j)+w_{(mk_I+j)}}$. Therefore, 
the parity-check equations of $\mathcal{C}_F$ take the form \vspace{-0.2cm}
\begin{equation*}\small
\begin{split}
\sum_{u=0}^{n_F-1}\left(\beta^{su+w_{(u)}}\right)^t c_{w,z,u}=0,&
\quad t\in[0,r_F-1],\\ w\in[0,&s^L-1],\;z\in[0,r_F-1],
\end{split}
\end{equation*}
where each coefficient $\beta^{su+w_{(u)}}$ depends on $w$ only 
through its $u$-th coordinate in the $s$-ary expansion. This is 
precisely the Hadamard MSR structure 
in~\eqref{eq: had_pc_row} and~\eqref{eq: had_row_points_simple}, with parameter $\alpha=(d_F-k_F+h)s^{n_F}$ and $s_F:=d_F-k_F+1=r_I+1$ dividing $s$ by definition. 
The multi-node repair scheme of Section~\ref{sec: hadamard}(ii) 
therefore applies with $h=r_F-r_I$ failed nodes and $d_F$ helper 
nodes.

\end{enumerate}

Once all written symbols are obtained, $\mathcal{C}_F$ supports 
both single-node and multiple-node repair. For single-node repair, 
fix any $z\in[0,r_F-1]$ and treat $\{c_{w,z,u}:w\in[0,s^L-1]\}$ 
as the subsymbols of node $u$. The repair argument follows the 
same argument as for $\mathcal{C}_I$ above with $d_F$ helper nodes, giving repair bandwidth 
$b=\alpha/s_F$. For multiple-node repair of $h=r_F-r_I$ 
simultaneously failed nodes, the repair procedure follows 
Section~\ref{sec: hadamard}(ii), with $d_F$ helper nodes. Hence 
$\mathcal{C}_F$ is an $(n_F,k_F;\alpha)_q$ MSR code.

\textbf{Bandwidth optimality.} 
The conversion bandwidth consists of 
two parts. The first stage reads all parity symbols from the $m$ 
initial codewords via $\sigma_1$, contributing $mr_I\alpha$ 
subsymbols. The second stage downloads $(r_F-r_I)k_Fs^L$ subsymbols 
from the unchanged nodes via the multi-node repair procedure. The 
total conversion bandwidth is therefore
\[\small
\gamma_R=mr_I\alpha+\frac{(r_F-r_I)k_F\alpha}{r_F}
=mk_I\alpha-mr_I\alpha(\frac{k_I}{r_F}-1),
\]
which meets the lower bound~\eqref{eq: acc+ban} with 
equality. Hence Construction~\ref{cons: array_bw} is bandwidth-optimal.
\end{IEEEproof}

\section{Conclusion}\label{sec: conclusion}

In this paper, we studied convertible MSR codes in the merge
regime and proposed explicit MSR-to-MSR conversion schemes
that jointly achieve efficient conversion and optimal repair.
Our constructions attain optimal access cost, optimal conversion
bandwidth, and the MSR repair property.
These results provide a unified approach to combining conversion
efficiency and repair optimality in distributed storage systems.

\bibliographystyle{IEEEtranS}
\bibliography{YBib}

\end{document}